\pdfoutput=1

\documentclass[11pt]{article}

\usepackage[latin9]{inputenc}
\usepackage[T1]{fontenc}
\usepackage{lmodern}

\usepackage{microtype}

\usepackage{amsmath}
\usepackage{amssymb}
\usepackage{amsthm}
\usepackage{thmtools}
\usepackage{thm-restate}

\usepackage{titlesec}
\usepackage{fancybox}
\usepackage{xcolor}
\usepackage{xspace}
\usepackage{enumerate}

\usepackage{comment}

\usepackage{bm}

\usepackage{fullpage}
\usepackage[procnumbered,ruled,vlined,linesnumbered]{algorithm2e}

\usepackage[colorlinks=true, linktocpage=true]{hyperref}

\newtheorem{theorem}{Theorem}[section]
\newtheorem{corollary}[theorem]{Corollary}
\newtheorem{lemma}[theorem]{Lemma}

\newtheorem{fact}[theorem]{Fact}
\newtheorem{assumption}[theorem]{Assumption}

\theoremstyle{definition}
\newtheorem{definition}[theorem]{Definition}

\newenvironment{fminipage}%
{\begin{Sbox}\begin{minipage}}%
		{\end{minipage}\end{Sbox}\fbox{\TheSbox}}

\def\defeq{\stackrel{\mathrm{def}}{=}}
\def\setof#1{\left\{#1  \right\}}
\def\sizeof#1{\left|#1  \right|}

\def\ceil#1{\left\lceil #1 \right\rceil}

\def\norm#1{\left\| #1 \right\|}

\newcommand{\SStil}{\boldsymbol{\mathit{\widetilde{S}}}}
\newcommand{\CCtil}{\boldsymbol{\mathit{\widetilde{C}}}}
\newcommand{\matlowtil}{\boldsymbol{\mathit{\widetilde{\mathcal{L}}}}}
\newcommand{\calDtil}{\boldsymbol{\mathit{\widetilde{\mathcal{D}}}}}

\newcommand{\cctil}{\boldsymbol{\mathit{\widetilde{c}}}}

\newcommand{\SShat}{\boldsymbol{\mathit{\widehat{S}}}}

\newcommand{\matlowhat}{\boldsymbol{\mathit{\widehat{\mathcal{L}}}}}
\newcommand{\calDhat}{\boldsymbol{\mathit{\widehat{\mathcal{D}}}}}

\newcommand{\calD}{\boldsymbol{\mathit{\mathcal{D}}}}

\newcommand{\matlow}{\boldsymbol{\mathit{{\mathcal{L}}}}}

\newcommand{\PP}{\boldsymbol{\mathit{P}}}

\newcommand{\rea}{\mathbb{R}}

\newcommand{\tauEst}{\widehat{\tau}}

\newcommand{\csamp}{\textsc{CliqueSample}}

\newcommand{\eps}{\epsilon}

\newcommand{\vecone}{\boldsymbol{\mathit{1}}}
\newcommand{\veczero}{\boldsymbol{\mathit{0}}}
\newcommand{\matzero}{\boldsymbol{\mathit{0}}}

\newcommand{\vstar}[2]{\ensuremath{\left(#1\right)_{#2}}}

\newcommand{\combLevEst}{\textsc{CombLevScoreEst}}
\newcommand{\combSparsify}{\textsc{CombGraphSparsify}}

\newcommand\exact{\textsc{SampleWithin}}
\newcommand\across{\textsc{SampleAcross}}
\newcommand\sample{\textsc{SampleEdge}}
\newcommand\helpestimatereff{\textsc{HelpEstimate}\text{$R_\text{eff}$}\xspace}
\newcommand\estimatereff{\textsc{Estimate}\text{$R_\text{eff}$}\xspace}
\newcommand\approxS{\textsc{ApproxSchur}}

\def\aa{\pmb{\mathit{a}}}
\newcommand\bb{\boldsymbol{\mathit{b}}}
\newcommand\cc{\boldsymbol{\mathit{c}}}

\newcommand\ww{\boldsymbol{\mathit{w}}}
\newcommand\yy{\boldsymbol{\mathit{y}}}

\newcommand\xx{\boldsymbol{\mathit{x}}}

\renewcommand\AA{\boldsymbol{\mathit{A}}}
\newcommand\BB{\boldsymbol{\mathit{B}}}
\newcommand\CC{\boldsymbol{\mathit{C}}}
\newcommand\DD{\boldsymbol{\mathit{D}}}

\newcommand\II{\boldsymbol{\mathit{I}}}

\newcommand\MM{\boldsymbol{\mathit{M}}}
\newcommand\LL{\bm{\mathit{L}}}
\newcommand\RR{\boldsymbol{\mathit{R}}}
\renewcommand\SS{\boldsymbol{\mathit{S}}}
\newcommand\TT{\boldsymbol{\mathit{T}}}

\newcommand\XX{\boldsymbol{\mathit{X}}}

\newcommand\LLtil{\boldsymbol{\mathit{\tilde{L}}}}

\newcommand\Otil{\widetilde{O}}

\makeatletter

\def\moverlay{\mathpalette\mov@rlay}
\def\mov@rlay#1#2{\leavevmode\vtop{%
		\baselineskip\z@skip \lineskiplimit-\maxdimen
		\ialign{\hfil$\m@th#1##$\hfil\cr#2\crcr}}}
\newcommand{\charfusion}[3][\mathord]{
	#1{\ifx#1\mathop\vphantom{#2}\fi
		\mathpalette\mov@rlay{#2\cr#3}
	}
	\ifx#1\mathop\expandafter\displaylimits\fi}
\makeatother

\colorlet{DarkRed}{red!50!black}
\colorlet{DarkGreen}{green!50!black}
\colorlet{DarkBlue}{blue!50!black}

\hypersetup{
	linkcolor = DarkRed,
	citecolor = DarkGreen,
	urlcolor = DarkBlue,
	bookmarksnumbered = true,
	linktocpage = true
}

\DontPrintSemicolon
\SetKw{KwAnd}{and}
\SetProcNameSty{textsc}
\SetFuncSty{textsc}
\SetKwInOut{Input}{Input}
\SetKwInOut{Output}{Output}

\let\oldnl\nl
\newcommand{\nonl}{\renewcommand{\nl}{\let\nl\oldnl}}

\ifdefined\ShowComment

\fi

\title{Sampling Random Spanning Trees \\ Faster than Matrix Multiplication}
\author{
David Durfee\thanks{Georgia Institute of Technology. \texttt{email:ddurfee@gatech.edu}}
\and
Rasmus Kyng\thanks{Yale University. \texttt{email:rasmus.kyng@yale.edu} }
\and
John Peebles\thanks{Massachusetts Institute of Technology. \texttt{email:jpeebles@mit.edu}}
\and
Anup B. Rao\thanks{Georgia Institute of Technology. \texttt{email:anup.rao@gatech.edu}}
\and
Sushant Sachdeva\thanks{Google. \texttt{email:sachdevasushant@gmail.com}. }
}

\hypersetup{
	pdftitle = {Sampling Random Spanning Trees Faster than Matrix Multiplication},
	pdfauthor = {David Durfee, Rasmus Kyng, John Peebles, Anup B. Rao, Sushant Sachdeva}
}

\date{}

\DeclareMathOperator{\diagop}{diag}
\newcommand{\scElim}{\approxS}
\newcommand{\combScElim}{\textsc{\textsc{CombApproxSchur}}}
\newcommand{\partialChol}{\textsc{ApxPartialCholesky}}
\newcommand{\levEst}{\textsc{LevScoreEst}}
\newcommand{\sparsify}{\textsc{GraphSparsify}}
\newcommand{\poly}{\mathrm{poly}}

\newcommand{\ulap}{\LL}
\newcommand{\schurto}[2]{\textsc{Schur}(#1,#2)}
\newcommand{\todolow}[1]{}

\begin{document}
\pagenumbering{roman}
\maketitle
\begin{abstract}
We present an algorithm that, with high probability, generates a random spanning tree from an edge-weighted undirected graph in $\Otil(n^{4/3}m^{1/2}+n^{2})$ time \footnote{The $\Otil(\cdot)$ notation hides $\poly(\log n)$ factors}. The tree is sampled from a distribution where the probability of each tree is proportional to the product of its edge weights. This improves upon the previous best algorithm due to Colbourn et al. that runs in matrix multiplication time, $O(n^\omega)$. For the special case of unweighted graphs, this improves upon the best previously known running time of $\tilde{O}(\min\{n^{\omega},m\sqrt{n},m^{4/3}\})$ for $m \gg n^{5/3}$ (Colbourn et al. '96, Kelner-Madry '09, Madry et al. '15).


The effective resistance metric is essential to our algorithm, as in the work of Madry et al., but we eschew determinant-based and random walk-based techniques used by previous algorithms. Instead, our algorithm is based on Gaussian elimination, and the fact that effective resistance is preserved in the graph resulting from eliminating a subset of vertices (called a Schur complement). As part of our algorithm, we show how to compute $\eps$-approximate effective resistances for a set $S$ of vertex pairs via approximate Schur complements in $\Otil(m+(n + |S|)\eps^{-2})$ time, without using the Johnson-Lindenstrauss lemma which requires $\Otil( \min\{(m + |S|)\eps^{-2}, m+n\eps^{-4} +|S|\eps^{-2}\})$ time. We combine this approximation procedure with an error correction procedure for handing edges where our estimate isn't sufficiently accurate.



\end{abstract}
\newpage

\setcounter{page}{0}
\pagenumbering{arabic}

\section{Introduction}
Random spanning trees are one of the most well-studied probabilistic
structures in graphs. Their history goes back to the classic
matrix-tree theorem due to Kirchoff in 1840s that connects the
spanning tree distribution to matrix
determinants~\cite{Kirchhoff47}. The task of algorithmically sampling
random spanning trees has been studied extensively~\cite{Guenoche83,
  Broder89, Aldous90,Kulkarni90,Wilson96,
  ColbournMN96,KelnerM09,MadryST15,HarveyX16}.

Over the past decade, sampling random spanning trees have found a few
surprising applications in theoretical computer science -- they were
at the core of the breakthroughs in approximating the traveling
salesman problem in both the symmetric~\cite{GharanSS11} and the asymmetric
case~\cite{AsadpourGMGS10}. Goyal et al.~\cite{GoyalRV09} showed that one could
construct a cut sparsifier by sampling random spanning trees.

Given an undirected, weighted graph $G(V,E,w),$ the algorithmic task
is to sample a tree with a probability that is proportional to the
product of the weights of the edges in the tree. We give an algorithm
for this problem, that, for a given $\delta >0,$ outputs a random
spanning tree from this distribution with probability $1-\delta$ in
expected time $\Otil((n^{4/3}m^{1/2}+n^{2})\log^4 1/\delta).$

For weighted graphs, a series of works building on the connection
between matrix trees and determinants, culminated in an algorithm due
to Colbourn, Myrvold, and Neufeld~\cite{ColbournMN96} that generates a
random spanning tree in matrix multiplication time
($O(n^{2.37..})$~\cite{Williams12}). Our result is the first
improvement on this bound for more than twenty years! It should be
emphasized that the applications to traveling salesman
problem~\cite{AsadpourGMGS10, GharanSS11} require sampling trees on
graphs with arbitrary weights.

A beautiful connection, independently discovered by
Broder~\cite{Broder89} and \cite{Aldous90} proved that one could
sample a random spanning tree, by simply taking a random walk in the
graph until it covers all nodes, and only keeping the first incoming
edge at each vertex. For graphs with unit-weight edges, this results
in an $O(mn)$ algorithm. The work of Kelner-Madry~\cite{KelnerM09} and
Madry et al.~\cite{MadryST15} are based on trying to speed up these
walks. These works together give a previously best running time of
$\tilde{O}(\min\{n^{\omega},m\sqrt{n},m^{4/3}\})$ for unit-weighted
graphs. Our algorithm is an improvement for all graphs with
$m \gtrsim n^{5/3}.$

The above works based on random walks seem challenging to generalize
to weighted graphs. The key challenge being that, in weighted graphs,
random walks can take a very long time to cover the graph. We take an
approach based on another intimate and beautiful connection; one
between random spanning trees and Laplacians.

\paragraph{Random Spanning Trees and the Laplacian Paradigm.}
A by-now well-known but beautiful fact states that the marginal
probability of an edge being in a random spanning tree is exactly
equal to the product of the edge weight and the effective resistance of the edge (see
Fact~\ref{fact:levprob}). Our algorithm will be roughly based on
estimating these marginals, and sampling edges accordingly. The key
challenge we overcome here, is that these sampling probabilities
change every time we condition on an edge being present or absent in
the tree.

Taking this approach of computing marginals allows us to utilize fast
Laplacian solvers and the extensive tools developed
therein~\cite{SpielmanTengSolver:journal, KoutisMP10:journal,
  KoutisMP11, KelnerOSZ13,lee2013efficient, CohenKMPPRX14, KyngLPSS16,
  KyngS16}. As part of our algorithm for generating random spanning
trees, we give a procedure to estimate all pair-wise effective
resistances in the graph without using the Johnson-Lindenstrauss
lemma. Our procedure is also faster if we only want to compute
effective resistances for a smaller subset of pairs.

Our procedure for estimating effective resistances is recursive. If we
focus on a small subset of the vertices, for the purpose of computing
effective resistances, we can eliminate the remaining vertices, and
compute the resulting \emph{Schur complement}. Computing the schur
complement exactly is costly and results in an $O(n^\omega)$ algorithm
(similar to~\cite{HarveyX16}). Instead, we develop a fast algorithm
for approximating the schur complement. Starting from a graph with $m$
edges, we can compute a schur complement onto $k$ vertices with at
most $\eps$ error (in the spectral sense), in $\Otil(m+n\eps^{-2})$
time. The resulting approximation has only $\Otil(k\eps^{-2})$ edges.

We hope that faster generation of random spanning trees and the tools
we develop here will find further applications, and become an integral
part of the \emph{Laplacian paradigm}.
\subsection{Prior Work}



One of the first major results in the study of spanning trees was Kirchoff's matrix-tree theorem, which states that the total number of spanning trees for general edge weighted graphs is equal to any cofactor of the associated graph Laplacian \cite{Kirchhoff47}.

Much of the earlier algorithmic study of random spanning trees heavily
utilized these determinant calculations by taking a random integer
between $1$ and the total number of trees, then efficiently mapping
the integer to a unique tree. This general technique was originally
used in \cite{Guenoche83, Kulkarni90}  to give an $O(mn^3)$-time
algorithm, and ultimately was improved to an $O(n^\omega)$-time
algorithm by \cite{ColbournMN96}, where $m,n$ are the numbers of edges
and vertices in the graph, respectively, and $\omega \approx 2.373$ is
the matrix multiplication exponent \cite{Williams12}. These
determinant-based algorithms have the advantage that they can handle
edge-weighted graphs, where the weight of a tree is defined as the
product of its edge weights.\footnote{To see why this definition is
  natural, note that this corresponds precisely to thinking of an edge
  with weight $k$ as representing $k$ parallel edges and then
  associating all spanning trees that differ only in which parallel
  edges they use.} Despite further improvements for unweighted graphs,
no algorithm prior to our work improved upon this $O(n^\omega)$
runtime in the general weighted case in over $20$ years since this
work. 
Even for unweighted graphs, nothing faster than $O(n^\omega)$ was known for dense graphs with $m \geq n^{1.78}$.

We now give a brief overview of the improvements for unweighted graphs along with a recent alternative $O(n^\omega)$ algorithm for weighted graphs.

Around the same time as the $O(n^\omega)$-time algorithm was discovered, Broder and Aldous independently showed that spanning trees could be randomly generated with random walks, where each time a new vertex is visited, the edge used to reach that vertex is added to the tree \cite{Broder89, Aldous90}. Accordingly, this results in an algorithm for generating random spanning trees that runs in the amount of time proportional to the time it takes for a random walk to cover the graph. For unweighted this \emph{cover time} is $O(mn)$ in expectation is better than $O(n^\omega)$ in sufficiently sparse graphs and worse in dense ones. However, in the more general case of edge-weighted graphs, the cover time can be exponential in the number of bits used to describe the weights. Thus, this algorithm does not yield any improvement in worst-case runtime for weighted graphs.
Wilson~\cite{Wilson96} gave an algorithm for generating a random spanning tree in expected time proportional to the mean hitting time in the graph.
This time is always upper bounded by the cover time, and it can be smaller.
As with cover time, in weighted graphs the mean hitting time can be exponential in the number of bits used to describe the weights, and so the this algorithm also does not yield an improvement in worst-case runtime for weighted graphs.

Kelner and Madry improved upon this result by showing how to simulate this random walk more efficiently. They observed that one does not need to simulate the portions of the walk that only visit previously visited vertices. Then, they use a low diameter decomposition of the graph to partition the graph into components that are covered quickly by the random walk and do precomputation to avoid explicitly simulating the random walk on each of these components after each is initially covered. This is done by calculating the probability that a random walk entering a component at each particular vertex exits on each particular vertex, which can be determined by solving Laplacian linear systems. This approach yields an expected runtime of $\widetilde{O}(m\sqrt{n})$ for unweighted graphs \cite{KelnerM09}.

This was subsequently improved for sufficiently sparse graphs with an algorithm that also uses shortcutting procedures to obtain an expected runtime of $\widetilde{O}(m^{4/3})$ in unweighted graphs \cite{MadryST15}. Their algorithm uses a new partition scheme based on effective resistance and additional shortcutting done by recursively finding trees on smaller graphs that correspond to random forests in the original graph, allowing the contraction and deletion of many edges. 

Recently, Harvey and Xu~\cite{HarveyX16} gave a simpler deterministic $O(n^\omega)$ time algorithm that uses conditional effective resistances to decide whether each edge is in the tree, contracting the edge in the graph if the edge will be in the tree and deleting the edge from the graph if the edge will not.\footnote{Note that for any edge $e$, there is a bijection between spanning trees of the graph in which $e$ is contracted and spanning trees of the original graph that contain $e$. Similarly, there is a bijection between spanning trees of the graph in which $e$ is deleted and spanning trees of the original graph that do not contain $e$.} Updating the effective resistance of each edge is done quickly by using recursive techniques similar to those in \cite{ColbournDN89} and via an extension of the Sherman-Morrison formula.



\section{Our Results}\label{sec:results}

\subsection{Random Spanning Trees}\label{sec:stResults}
\begin{restatable}[]{theorem}{mainTrees}
	\label{thm:mainTrees}
For any $0< \delta <1,$ the routine $\textsc{GenerateSpanningTree}$ (Algorithm~\ref{alg:mainAlgo})   outputs a random spanning tree from the $\ww$-uniform distribution  with probability at least $1 -\delta$ and takes expected time $\Otil((n^{4/3}m^{1/2}+n^{2}) \log^4 1/\delta).$
\end{restatable}

\

Our algorithm samples edges according to their conditional effective resistance as in \cite{HarveyX16}. We repeatedly use the well known fact that the effective resistance multiplied by the edge weight, which we will refer to as the \textit{leverage score} of the edge,  is equal to  the probability that the edge belongs to a randomly generated spanning tree.  To generate a uniformly random spanning tree, one can sample edges in an iterative fashion. In every iteration, the edge being considered is added to the spanning tree with probability exactly equal to its leverage score. If it is added to the tree, the graph is updated by contracting that edge, otherwise, the edge is removed from the graph.
Though using fast Laplacian solvers~\cite{SpielmanTengSolver:journal} one can compute the leverage score of a single edge in $\Otil(m)$ time, since one needs to potentially do this $m$ times (and the graph keeps changing every iteration), this can take $\Otil(m^2)$ time if done in a naive way.  It therefore becomes necessary to compute the leverage scores in a more clever manner.

The algorithms in \cite{ColbournDN89,HarveyX16} get a speed up by a clever recursive structure which enables one to work with much smaller graphs to compute leverage scores at the cost of building such a structure. This kind of recursion will be  the starting point of our algorithm which  will randomly partition the vertices into two equally sized sets, and compute Schur complements onto each of the set. We crucially use the fact that Schur complement, which can be viewed as block Gaussian elimination, preserves effective resistances of all the edges whose incident vertices are not eliminated. We first recursively sample edges contained in both these sets, contracting or deleting every edge along the way,  and then the edges that go across the partition is sampled. Algorithm in \cite{HarveyX16} is essentially this, and they prove that it takes $O(n^\omega).$

In order to improve the running time, the main workhorse we use is derived from the recent paper \cite{KyngS16} on fast Laplacian solvers which provided an almost linear time algorithm for performing an approximate Gaussian elimination of Laplacians. We generalize the statement in \cite{KyngS16} to show that one can compute an approximate Schur complement of a set of vertices in a Laplacian quickly. Accordingly, one of our primary results, discussed in Section~\ref{sec:apxSchurResults}, will be that a spectrally approximate Schur complement can be efficiently computed, and we will leverage this result to achieve a faster algorithm for generating random spanning trees.

Since we compute approximate Schur complements, the leverage scores of edges are preserved only approximately. But we set the error parameter such that we can get a better estimate of the leverage score if we move up the recursion tree, at the cost of paying more for the computing leverage score of an edge in a bigger graph. We give a sampling procedure that samples edges into the random spanning tree from the true distribution by showing that approximate leverage score can be used to make the right decisions most of the times. 

Subsequently, we are presented with a natural trade-off for our error parameter choice in the $\approxS$ routine: larger errors speed up the runtime of $\approxS$, but smaller errors make moving up the recursion to obtain a more exact effective resistance estimate less likely. Furthermore, the recursive construction will cause the total vertices across each level to double making small error parameters even more costly as we recurse down. Our choice of the error parameter will balance these trade-offs to optimize running time.

The routine $\approxS$ produced an approximate Schur complement only with high probability. We are not aware of a way to certify that a graph sparsifier is good quickly. Therefore, we condition on the event that the $\approxS$ produces correct output on all the calls, and show ultimately show that it is true with high probability.

Our algorithm for approximately generating random spanning trees, along with a proof of Theorem~\ref{thm:mainTrees} is given in Section ~\ref{sec:algorithm} and ~\ref{sec:runtime}

\subsection{Approximating the Schur complement}\label{sec:apxSchurResults}

\begin{restatable}[]{theorem}{schurApx}
	\label{thm:schurApx}
  Given a connected undirected multi-graph
  $G =(V,E)$, with positive edges weights 
  $w : E \to \rea_{+}$, and associated Laplacian $\LL$,
  a set vertices $C \subset V$,
  and scalars $0<\eps\leq1/2$, $0< \delta < 1$,
  the algorithm $\scElim(\LL,C,\eps,\delta)$
  returns a Laplacian matrix $\SStil$.
  With probability $\geq 1-\delta$, the following statements all hold:
  $\SStil \approx_{\eps} \SS$, where $\SS$ is the Schur complement of
  $\LL$ w.r.t elimination of $F=V-C$.
  $\SStil$ is a Laplacian matrix whose edges are supported on $C$.
  Let $k = \sizeof{C} = n-\sizeof{F}$.
  The total number of non-zero entries $\SStil$ is $O(k
  \eps^{-2}\log(n/\delta))$.
  The total running time is bounded by
  $O((m\log n \log^{2}(n/\delta)+
  n \eps^{-2} \log n \log^{4} (n/\delta))
  \operatorname{polyloglog}(n))$.
\end{restatable}
The proof of this appears in Section~\ref{sec:schurapx}.

As indicated earlier, the algorithm $\approxS$ is builds on the tools developed in~\cite{KyngS16}.
Roughly speaking, the algorithm in~\cite{KyngS16} produces an $\epsilon$-approximation to a Cholesky decomposition of the Laplacian in 
$\Otil(\frac{m}{\epsilon^2})$ time.
Our algorithm for approximating Schur complements is based on three key modifications to the algorithm from~\cite{KyngS16}:
Firstly, we show that the algorithm can be used to eliminate an arbitrary subset $U \subset V$ of the vertices, giving a approximate partial Cholesky decomposition.
Part of this decomposition is an approximate Schur complement w.r.t. elimination of the set of vertices $U$. 
Secondly, we show that although the spectral approximation quality of this decomposition is measured in terms of the whole Laplacian, in fact it implies a seemingly stronger guarantee on the approximate Schur complement: Its quadratic form resembles the true Schur complement up to a small multiplicative error.
Thirdly, we show that the algorithm from~\cite{KyngS16} can utilize leverage score estimates (constant-factor approximations) to produce an approximation in only $\Otil(m + {n}{\epsilon^{-2}})$ time. 
Additionally, we also sparsify the output to ensure that the final approximation has only $\Otil((n-|U|)\epsilon^{-2})$ edges. The leverage score estimates can be obtained by combining a Laplacian solver with Johnson-Lindenstrauss projection. 
It is worth noting that the Laplacian solver from~\cite{KyngLPSS16} is also based on approximating Schur complements. However, their algorithm can only approximate the Schur complement obtained by eliminating very special subsets of vertices. The above theorem, in contrast, applies to an arbitrary set of vertices.
This algorithm also had a much worse dependence on $\eps^{-1}$, making it unsuitable for our applications where $\eps^{-1}$ is $\Omega(n^{c})$ for some small constant $c$.


\subsection{Computing Effective Resistance}\label{sec:erResults}

Our techniques also enable us to develop a novel algorithm for computing the effective resistances of pairs of vertices. In contrast with prior work, our algorithm does not rely on the Johnson-Lindenstrauss lemma, and it achieves asymptotically faster running times in certain parameter regimes.

\begin{restatable}[]{theorem}{estimateReffAnalysis}
	\label{thm:esimateReff-analysis}
When given a graph $G$, a set $S$ of pairs of vertices, and an error parameter $\epsilon$, the function $\estimatereff(G,S,\epsilon)$ (Algorithm~\ref{alg:estimateReff} in Section~\ref{sec:reff-estimation}) returns $e^{\pm \epsilon}$-multiplicative estimates of the effective resistance of each of the pairs in $S$ in time $\widetilde{O}\left(m + \frac{n+|S|}{\epsilon^2}\right)$ with high probability.
\end{restatable}

One can compare this runtime with what can be obtained using (now standard) linear system solving machinery introduced in \cite{SpielmanTengSolver:journal}. Using such machinery, one obtains an algorithm for this same problem with runtime\footnote{The first runtime in the min expression comes from applying JL with the original Laplacian. The second runtime comes from sparsifying the Laplacian first and then applying JL.} $\widetilde{O}\left((\min\left(\frac{m+|S|}{\epsilon^2},m+\frac{n}{\epsilon^4}+\frac{|S|}{\epsilon^2}\right)\right)$. When the number of pairs $|S|$ and the error parameter $\epsilon$ are both small, the runtime of our algorithm is asymptotically smaller than this existing work.


\section{Notation}	

\subsection*{Graphs}
We assume we are given a weighted undirected graph $G = (V,E,\ww),$ with the vertices are labelled  $V = \{ 1, 2,...,n \}.$ Let $\AA_G$ be its adjacency matrix. The $(i,j)$'th entry of the adjacency matrix $\AA_G(i,j) = \ww_{i,j}$ is the weight of the edge between the vertices $i$ and $j$.  Let $\DD_G$ is the diagonal matrix consisting of degrees of the vertices, i.e., $\DD_G(i,i) = \text{deg}_G(i).$ The Laplacian matrix is defined as $\LL_G = \DD_G - \AA_G.$ We drop the subscript $G$ when the underlying graph is clear from the discussion.

\begin{definition}[Induced Graph]
Given a graph $G = (V,E)$ and a set of vertices $V_1 \subseteq V,$ we use the notation $G(V_1)$ to mean the induced graph on $V_1.$
\end{definition}

\begin{definition}
Given a set of edges $E$ on vertices $V$, and $V_1, V_2 \subseteq V$, we use the notation $E \cap (V_1,V_2)$ to mean the set of all edges in $E$ with one end point in $V_1$ and the other in $V_2.$ 
\end{definition}

\begin{definition}[Contraction and Deletion]
Given a graph $G = (V,E)$ and a set of edges $E_1 \subset E$, we use the notation $G \backslash E_1$ to denote the graph obtained by deleting the edges in $E_1$ from $G$ and $G/E_1$ to denote the graph obtained by contracting the edges in $E_1$ within $G$ and deleting all the self loops.
\end{definition}

\subsection*{Spanning Trees}

Let $\mathcal{T}_G$ denote the set  of all spanning subtrees of $G.$ We now define a probability distribution on these trees.

\begin{definition}[$\ww$-uniform distribution on trees]
Let $\mathcal{D}_{G}$ be a probability distribution on $\mathcal{T}_G$ such that
$$ \Pr \left( X = T | \; X \sim \mathcal{D}_{G} \right) \propto \Pi_{e \in T} \ww_e.$$
We refer to $\mathcal{D}_{G}$ as the $\ww$-uniform distribution on $\mathcal{T}_G$. When the graph $G$ is unweighted, this corresponds to the uniform distribution on $\mathcal{T}_G.$
\end{definition}

\begin{definition}[Effective Resistance]
	The \textit{effective resistance} of a pair of vertices $u,v \in V_G$ is defined as 
	$$ R_{eff}(u,v) = \bb_{u,v}^T \LL^{\dagger} \bb_{u,v}.$$ where $\bb_{u,v}$ is an all zero vector corresponding to $V_G$, except for entries of 1 at $u$ and $v$
\end{definition}

\begin{definition}[Levarage Score]
  The \textit{statistical leverage score}, which we will abbreviate to \textit{leverage score}, of an edge $e = (u,v) \in E_G$ is defined as 
$$ l_e = \ww_e R_{eff}(u,v).$$
\end{definition}

\begin{fact}[Spanning Tree Marginals]\label{fact:levprob}
The probability $\Pr(e)$ that an edge $e \in E_G$ appears in a tree sampled $\ww$-uniformly randomly from  $\mathcal{T}_G$ is given by 
$$ \Pr(e) = l_e,$$
where $l_e$ is the leverage score of the edge $e.$
\end{fact}

\subsection*{Schur Complement}

\begin{definition}[Schur Complement]\label{def:schurcomp}
Let $\MM$ be a block matrix
	
		\begin{align}
		\MM 
		& =
		\left[
		\begin{array}{cc}
		\AA & \BB \\
		\BB^T & \CC
		\end{array}
		\right].
		\end{align}	
We use $\textsc{Schur}(\MM,\AA)$ to denote the Schur complement of $\CC$ onto $\AA$ in $\MM$; ie., \[\textsc{Schur}(\MM,\AA) = \AA - \BB \CC^{-1}\BB^T.\] Equivalently, this is simply the result of running Gaussian elimination of the block $\CC$.
\end{definition}
When the matrix $\MM = \LL$ is a Laplacian of a graph $G=(V,E)$ and $V_1 \subseteq V$ is a set of vertices, we abuse the notaion and use $\textsc{Schur}(\LL,V_1)$ or $\textsc{Schur}(G,V_1)$  to denote the Schur complement of $\LL$ onto the submatrix of $\LL$ corresponding to $V_1$; i.e., onto the submatrix of $\LL$ consisting of all entries whose coordinates $(i,j)$ satisfy $i,j \in V_1$.

\begin{fact}
Let $G = (V,E)$ be a graph and $V = V_1 \cup V_2$ be a partition of the vertices. Then $\textsc{Schur}(G,V_1)$ is Laplacian matrix of a graph on vertices in $V_1$. 
\end{fact}
This means that Schur complement in a graph $G=(V,E)$ onto a set of vertices $V_1$ can be viewed as a graph on $V_1$. Furthermore, we can view this as a multigraph obtained by adding (potentially parallel) edges to $G(V_1)$, the induced graph on $V_1.$ We take this view in this paper: whenever we talk about Schur complements, we separate out the edges of the original graph from the ones created during Schur complement operation. 

We now provide some basic facts about how Schur complements relate to spanning trees. This first lemma says that edge deletions and contractions commute with taking Schur complements.

\begin{fact} 
\label{fact:schurConDel}	
		(Lemma 4.1 of \cite{ColbournDN89}) Given $G$ with any vertex partition $V_1,V_2$, for any edge $e \in E \cap (V_1,V_1)$.
		 \[\textsc{Schur}(G\setminus e,V_1) = \textsc{Schur}(G,V_1)\setminus e \qquad \text{and} \qquad \textsc{Schur}(G/e,V_1) = \textsc{Schur}(G,V_1)/e\] 
\end{fact}

\begin{fact}
\label{fact:schurLeverage}
	Given $G$ with any vertex partition $V_1,V_2$, for any edge $e \in E \cap (V_1,V_1)$, the leverage score of $e$  in $G$ is same as that in $\textsc{Schur}(G,V_1)$.
	
\end{fact}

\begin{proof}
	This follows immediately from Fact~\ref{fact:schurConDel}, Kirchhoff's matrix-tree theorem \cite{Kirchhoff47}, and the fact that Gaussian elimination preserves determinant.
	
\end{proof}

\subsection*{Spectral Approximation}

\begin{definition}
	Given two graphs $G,H$ on identical vertex sets, and respective Laplacians $L_G$ and $L_H$. We say $G \approx_{\epsilon} H$ if 
	$$\exp(-\epsilon) L_H \preceq  L_G \preceq  \exp(\epsilon)  L_H.$$
\end{definition}

\begin{definition}[Approximate Schur Complement]\label{def:appschurComp}
	Given a graph $G=(V,E)$ and vertex set $U \subset V$, let $S_U$ be the Laplacian of $\textsc{Schur}(G,V\setminus U) - G(V\setminus U)$; ie., the set of edges added to the induced subgraph $G(V\setminus U)$ by the Schur complement operation. 
	We call a matrix $\widetilde{S}_{U}$ an $\epsilon$-\textit{approximate Schur complement} if it satisfies 
	$$  \widetilde{S}_U \approx_{\epsilon} S_U.$$
	Furthermore,  $\widetilde{S}_{U}$ is a Laplacian.
\end{definition}

\section{Algorithm for Sampling Spanning Trees}\label{sec:algorithm}

\newcommand\Gtil{\widetilde{G}}

It is well known that for any edge of a graph, the probability of that edge appearing in a random spanning tree is equal to it's leverage score. We can iteratively apply this fact to sample a $\ww$-uniform random tree. We can consider the edges in an arbitrary sequential order, say $e_1,...,e_m \in E$, and make decisions on whether they belong to tree. Having decided for edges $e_1,...,e_{i}$, one computes the probability $p_{i+1}$,  conditional on the previous decisions, that edge $e_{i+1}$  belongs to the tree. Edge $e_{i+1}$ is then added  to the tree with probability  $p_{i+1}.$ 

To estimate the probability that edge $e_{i+1}$ belongs to the tree conditional on the decisions made on $e_1,...,e_i$, we can use Fact~\ref{fact:levprob}. Let $E_{T} \subset \{e_1,...,e_i\} $ be the set of edges that were included in the tree, and $F_{T} \subset \{e_1,...,e_i\} $ the subset of edges that were not included. Then, $p_{i+1}$ is equal to the leverage score of  edge $e_{i+1}$ in the graph $G^{(i+1)} := \left( G \backslash F_{T} \right)/E_T$ obtained by deleting edges $E_T$ from $G$ and then contracting edges $E_{T^c}.$   In other words, we get $G^{(i+1)}$ from $G^{(i)}$ by either deleting the edge $e_i$  or contracting it, depending on if $e_i$ was not added to the tree or added to the tree, respectively. Note that as we move along the sequence, some of the original edges may no longer exist in the updated graph due to edge contractions. In that case, we just skip the edge and move to the next one.
 
 Computing leverage score of an edge, with $\epsilon$ multiplicative error, requires $\Otil(m\log{1/\epsilon})$ runtime. Since we potentially have to compute leverage score of every edge, this immediately gives a total runtime of $\Otil(m^2)$.
 
  Our algorithm will similarly make decisions on edges in a sequential order. Where it differs from the above algorithm is the graph we use to compute the leverage score of the edge. Instead of computing the leverage score of an edge in the original graph updated with appropriate contractions and deletions, we deal with potentially much smaller graphs containing the edge such that the effective resistance of the edge in the smaller graph is approximately same as in the original graph. In the next section, we describe the sampling procedure  that we use to sample from the true distribution, when we have access to a cheap but approximate routine to compute the sampling probability.
  
\subsection{Structure of the Recursion}
We now describe the recursive structure of the algorithm given in Algorithm~\ref{alg:mainAlgo}. The structure of the recursion is same as in \cite{HarveyX16}. Let the input graph be $G = (V_G,E_G)$. Suppose at some stage of the algorithm, we have a graph $\widetilde{G}. $ The task is to make decisions on edges in $E_G \cap E_{\Gtil}.$ We initially divide the vertex set into two equal sized sets $V_{\Gtil} = V_1 \cup V_2.$ Recursively, we first make decisions on edges in $\Gtil(V_1) \cap E_G$, then make decisions on edges in $\Gtil(V_2) \cap E_G$ and finally make decisions on the remaining edges. To make decisions on $\Gtil(V_1) \cap E_G,$ we use the fact that the effective resistance of edges are preserved under Schur complement. We work with the graph $G_1 = \approxS(\Gtil,V_1,\epsilon)$ and recursively make decisions on edges in $E_G \cap G(V_1).$ Having recursively made decisions on edges in $E_G \cap \Gtil(V_1)$, let $E_T$ be the set of tree edges from this set. We now need to update the graph $\Gtil$ by contracting edges in $E_T$ and deleting all the edges in $E_T^c \cap \Gtil(V_1) \cap E_G.$ Then we do the same for the edges in $E_G \cap \Gtil(V_2).$ 

Finally, we treat the edges  $E_G \cap (V_1,V_2)$ that cross $V_1,V_2$ in a slightly different way, and is handled by the subroutine $\across$ in the algorithm. If we just consider the edges in $E_G$, this is trivially a bipartite graph. This property is maintained in all the recursive calls by the routine $\across.$ The routine $\across$ works by dividing $V_1,V_2$ both into two equal sized sets $V_1 = L_1 \cup L_2$ and $V_2 = R_1 \cup R_2$ and making four recursive calls, one each  for edges in $E_G \cap (L_i,R_j), i=1,2; j=1,2.$ To make decisions on edges in $E_G \cap (L_i,R_j),$ it recursively calls $\across$ on the graph $G_{ij} = \approxS(\Gtil, (L_i, R_j)^c, \epsilon)$ obtained by computing approximate Schur complement on to vertices in $(L_i,R_j)$ of vertices outside it. 

\subsubsection{Exact Schur Complement and $O(n^{\omega})$ Time Algorithm}
Here we  note how we can get a $O(n^{\omega})$ algorithm. Note that this is very similar to the algorithm and analysis in \cite{HarveyX16}. If in $\approxS$ calls, we set $\epsilon = 0$, i.e., we compute exact Schur complements, then we have a $O(n^{\omega})$ algorithm. Whenever we make a decision on an edge by instantiating $\sample(e)$, we just have to compute the leverage score $l_e$ of the edge $e$ in a constant sized graph. This can be done in constant time and since we do exact Schur complements, $l_e = l_e(G).$ We can therefore use this to decide if $e$ belongs to the tree and then update the graph by either contracting the edge or deleting it depending on if it is included or excluded in the tree. In a graph with $n_1$ vertices, it takes $O(n_1^{\omega})$ time to compute the Schur complement. Let $T(n)$ be the time taken by $\exact$ on a graph of size $n$ and $B(n)$ be the time taken by $\across$ when called on a graph of size $n$. We then have the following recursion
\begin{align*}
T(n) &= 2 T(n/2) + B(n) + O(n^{\omega}) \\
B(n) &= 4 B(n/2) + O(n^{\omega}).
\end{align*}
We therefore have $T(n) = O(n^{\omega}).$

\subsubsection{Approximate Schur Complement and Expected $\Otil(n^{4/3}m^{1/2}+n^{2})$ Time Algorithm}
We speed up $O(n^{\omega})$ algorithm by computing approximate Schur complements faster.  Having access only to approximate Schur complements,  which preserves leverage score only approximately, introduces an issue with computing sampling probability. It is a-priori not clear how to make decisions on edges when we preserve leverage scores only approximately during the recursive calls.  The key idea here is as follows. Suppose we want to decide if a particular edge $e$ belongs to the tree. Tracing the recursion tree produced by Algorithm~\ref{alg:mainAlgo}, we see that we have a sequence of graphs $G ,G_1,G_2,...,G_k$ all containing the edge $e,$ starting from the original input graph $G$ all the way down to $G_k$ which has a constant number of vertices. We also have $V(G_i) \subset V(G_{i-1})$ for all $k \geq i \geq 1,$ all of them being subsets of $V(G).$  

Let $n = |V(G)|, m = |E_G|$ be the number of vertices and edges in the input graph, 
When setting the error parameters, we choose $\epsilon$ and some threshold values in ways that 
depend on whether $m \leq n^{4/3}$ holds.
In the case $m > n^{4/3}$, 
we define $\epsilon$ in terms of the level $i$ as
\begin{equation}
\label{eq:errordense}
\epsilon(i) = 2^{i/2} n^{-1/6} m^{-1/4} \log^{-2} n 
.
\end{equation}
In the case $m \leq n^{4/3}$, 
we define $\epsilon$ in terms of the level $i$ as
\begin{equation}
\label{eq:errorsparse}
\epsilon(i) = 2^{i/2} n^{-1/2} \log^{-2} n 
.
\end{equation}
%
The threshold value is $t_1$ is such that $2^{2 t_1} = \frac{n^2}{m}$.

Our $\epsilon(\cdot)$ function will ensures for all $i,$  $l_e(G) \in [(1- \epsilon_i)l_e(G_i), (1 + \epsilon_i)l_e(G_i) ]$ for an appropriate $\epsilon_i.$ We sample a uniform random number $r \in [0,1],$ and initially compute $l_e(G_k).$   If $r$ lies outside the interval $ [(1- \epsilon_i)l_e(G_k), (1 + \epsilon_i)l_e(G_k) ],$ then we can make a decision on the edge $e$. Otherwise, we estimate $l_e(G)$ to a higher accuracy by computing $l_e(G_{k-1})$. We continue this way, and if $r$ lies inside the interval  $ [(1- \epsilon_i)l_e(G_i), (1 + \epsilon_i)l_e(G_i) ]$ for every $i$, then we compute $l_e(G)$ in the input graph $G.$ In the next section we describe $\sample$ in more detail.

At this point, we find it important to mention that the spectral error guarantees from the $\approxS$ subroutine only hold with probability $\geq 1 - O(\delta)$. The explanation of the $\sample$ subroutine above relied on these spectral guarantees, and the error in our algorithm for generating random spanning trees will be entirely due to situations in which the sparsification routine does not give a spectrally similar Schur complement. For the time being we will work under the following assumption and later use the fact that it is true w.h.p. to bound the error of our algorithm.

\begin{assumption}\label{assumption}
	Every call to $\approxS$ with error parameter $\epsilon$
        always computes an $\epsilon$-approximate Schur Complement.
\end{assumption}

\begin{algorithm}								
\caption{$\textsc{GenerateSpanningTree}(\widetilde{G}=(E_{\widetilde{G}},\widetilde{V})):$ Recurse using Schur Complement}
\label{alg:mainAlgo}
\SetAlgoVlined

\KwIn{Graph $\widetilde{G}$. Let $E_G,$ a global variable, denote the edges in the original (input) graph $G$.}
\KwOut{$E_T$ is the set of edges in the sampled tree.}						

$E_T \leftarrow \exact(G)$ \;
 \KwRet {$E_T$}\;
 
\textbf{Procedure }{$\exact(\widetilde{G})$}{

Set $E_T \leftarrow \{ \}$ \;
 \eIf {$| \widetilde{V}| = 1  $} { 
 \KwRet {}\;
}
 {
 Divide $V$ into equal sets  $V = V_1 \cup V_2 .$ \;
\For{$i=1,2$}{ 
Compute $G_i = \approxS(\widetilde{G},V_{i}, \epsilon( \text{level}))$
(see Equations~\eqref{eq:errorsparse} and \eqref{eq:errordense})  \;
 $E_T \leftarrow E_T \cup \exact(G_i)$ \;
Update $\widetilde{G}$ by deleting edges in $ \widetilde{G}(V_i) \cap E_T^c$ and contracting edges in $\widetilde{G}(V_i) \cap E_T.$ (Note the convention $E_T^c := E_G \backslash E_T $ ) \;
}
$E_T \leftarrow E_T \cup \across(\widetilde{G},(V_1,V_2))$ \; 
\KwRet{$E_T$} \;
}
}

\textbf{Procedure }{$\across(\widetilde{G},(L,R))$}{
\If{$|L| = |R| = 1$}{

$E_T = \sample(\widetilde{G}, (L, R) \cap E_G )$ \;
\KwRet{$E_T$} \;
}

 Divide $L,R$ into two equal sized sets: $L = L_1 \cup L_2,$ $R = R_1 \cup R_2.$ \;
\For{$i=1,2$}{
\For{$j=1,2$}{
$\widetilde{G}_{ij} \leftarrow \approxS(\widetilde{G},(L_i \cup R_j),
\epsilon(\text{level}))$ (see Equations~\eqref{eq:errorsparse} and \eqref{eq:errordense}) \;
$E_T \leftarrow E_T \cup \across(\widetilde{G}_{ij},(L_i,R_j))$ \;
Update $\widetilde{G}$ by contracting edges $E_T$ and deleting edges in $E_T^c \cap (L_i,R_j)$ \;

}

}
\KwRet{$E_T$} \;
  }

\end{algorithm}

  \subsubsection*{Sampling Scheme: $\sample$}
 
In this section we describe the routine $\sample(e)$ for an edge $e \in G$ in the input graph. By keeping track of the recursion tree, we have $G_0, G_1, ...,G_k$ and $e \in G_i$ for all $i$.

\begin{lemma}\label{lem:approx_leverage}
	For graph $G$ and $G_i$, the respective conditional leverage scores $l_e$ and $l_e^{(i)}$ for edge $e$ are such that $l_e \in [(1 - 2 \epsilon(i) \log n)l_e^{(i)},(1 + 2 \epsilon(i) \log n)l_e^{(i)} ]$
	
\end{lemma}

This will now allow us to set $\epsilon_i = 2 \epsilon(i) \log{n}$. We will delay the proof of Lemma~\ref{lem:approx_leverage} until later in this section in favor of first giving the sampling procedure. The sampling procedure is as follows. We generate a uniform random number in $r \in [0,1].$ We want to sample edge $e$ if $r \leq l_e(G).$ Instead, we use $l_e(G_k)$ as a proxy. Note that using fast Laplacian solvers, we can in $\Otil(\text{no. of edges})$ time compute leverage score of an edge upto a factor of $1 + 1/\text{poly}(n).$   Since $l_e(G) \in [(1-\epsilon_k )l_e(G_k), (1+\epsilon_k)l_e(G_k)],$ we include the edge in the tree if $r \leq  (1-\epsilon_k)l_e(G_k),$ otherwise if $r > (1+\epsilon_k)l_e(G_k) ,$ we don't include it in the tree. If $r \in  [(1-\epsilon_k)l_e(G_k), (1+\epsilon_k)l_e(G_k)],$ which happens with probability $2\epsilon_k l_e(G_k),$ we get a better estimate of $l_e(G)$ by computing $l_e(G_{k-1}).$ We can make a decision as long as $r \notin  [(1-\epsilon_{k-1})l_e(G_{k-1}), (1+\epsilon_{k-1})l_e(G_{k-1})] ,$ otherwise, we consider the bigger graph $G_{k-2}$. In general, if $r \notin  [(1-\epsilon_i)l_e(G_i), (1+\epsilon_i)l_e(G_i)],$  then we can make a decision on $e,$ otherwise we get a better approximation of $l_e(G)$ by computing $l_e(G_{i-1}).$ If we can't make a decision in any of the $k$ steps, which happens if $r \in  [(1-\epsilon_i)l_e(G_i), (1+\epsilon_i)l_e(G_i)]$ for all $i$, then we compute the leverage score of $e$ in $G$ updated with edge deletions and contractions resulting from decisions made on all the edges that were considered before $e$. 

Note that when we fail to get a good estimate at level $i$ for some $i \geq t_1$, we always compute the next estimate with respect to the original graph.


Finally, note that in the final step, we can compute $l_e(G)$  up to an approximation factor of $1 + \rho$ in $\Otil(m \log 1/\rho).$ We can therefore start with $\delta_0 = 1/n$ and if $r \in [(1-\rho)\widetilde{l}_e(G), (1+\rho)\widetilde{l}_e(G)] ,$ we set $\rho = \rho_0/2$ and repeat.  This terminates in $\Otil(m)$ expected (over randomness in $r$) time. 

For our algorithm, assume that we have an efficient data structure that gives access to each graph $G_0,...G_k$ in which $e$ appears.

	

	
	


	
	
		




\begin{algorithm}								
	\caption{$\sample(e):$ Sample an edge using conditional leverage score}
	\label{alg:edgeSample}
	\SetAlgoVlined
	

	\KwIn{An edge $e$ and access to graphs $G_0,...G_k$ in which $e$ appears }
	\KwOut{Returns $\{ e \}$ if edge belongs to the tree, and $\{ \}$ if  it doesn't }						
	Generate a uniform random number $r$ in $[0,1]$ \;
	
	$l_e \leftarrow \textsc{EstimateLeverageScore}(e)$
	
	\eIf{$r < l_e $}{
	\KwRet{$\{e \}$}
	}{
	\KwRet{$\{ \}$}
	}


	
	\textbf{Procedure }{$\textsc{EstimateLeverageScore}(e)$}{
	Compute $l_e^{(k)}$ to error $1/n$ \;
	\If{$\textsc{isGood}(l_e^{(k)}, \epsilon)$}{
	\KwRet{$l_e^{(k)}$}
	}
	
	\For{$i=t_1$ to $\log n$}{ 
			Compute $l$, an estimate for $l_e^{(i)}$ with error $1/n$\;
			\If{$\textsc{isGood}(l, \epsilon(i))$}{
			\KwRet{$l$}
		}
		}
   		
		\For{$i=0$ to $\infty$}{ 
			Compute $l$, an estimate for $l_e^{(0)}$ with error $2^{-i}n$\;
			\If{$\textsc{isGood}(l, 2^{-i} n)$}{
			\KwRet{$l$} }
		}

	}
	\textbf{Procedure }{$\textsc{isGood}(l_e, \epsilon)$}{
		\If{$r < (1 - \epsilon)l_e \text{ or } r > (1 + \epsilon)l_e$}{

			\KwRet{True}			
		}
		\KwRet{False}
%
	}

\end{algorithm}

\subsubsection*{Proof of Lemma~\ref{lem:approx_leverage}}
This edge sampling scheme relies upon the error in the leverage score estimates remaining small as we work our way down the subgraphs and remaining small when we contract and delete edges. Theorem~\ref{thm:schurApx} implies leverage score estimates will have small error between levels, so we will only have compounding of small errors. However, it does not imply that these errors remain small after edge contractions and deletions, which becomes necessary to prove in the following lemma.
\begin{lemma}\label{lem:approx_cond_lev}
		Given a graph $G = (V,E)$, vertex partition $V_1,V_2$, and edges $e \in E \cap (V_1,V_1)$, then 
		\[\approxS(G,V_1,\epsilon)/e \approx_{\epsilon} \schurto{G/e}{V_1}, \approxS(G,V_1,\epsilon)\setminus e \approx_{\epsilon} \schurto{G\setminus e}{V_1}\]
	
\end{lemma}

\begin{proof}
$\approxS(G,V_1,\epsilon)/e \approx_{\epsilon} \schurto{G}{V_1}
/e$ because spectral approximations are maintained under contractions. Furthermore, $\approxS(G,V_1,\epsilon) = \ulap_{V_1} + \tilde{S_{V_2}}$ where $\ulap_{V_1}$ is the Laplacian of the edges in $E \cap (V_1,V_1)$. Similarly, write $\schurto{G}{V_1} = \ulap_{V_1} + {\SS_{V_2}}$, and because $\tilde{\SS}_{V_2} \approx_{\epsilon} \SS_{V_2}$ then $\ulap_{V_1} \setminus e + \tilde{\SS}_{V_2} \approx_\epsilon \ulap_{V_1} \setminus e + {\SS_{V_2}}$. Combining these facts with Fact~\ref{fact:schurConDel} gives the desired result.  
\end{proof}

\begin{proof} (of Lemma~\ref{lem:approx_leverage})

	By construction, $\epsilon(i) \leq \epsilon(k)$ for every $i \leq k$. Iteratively applying Theorem~\ref{thm:schurApx} and Lemma~\ref{lem:approx_cond_lev}, gives $l_e \in [e^{-\epsilon(k)k}{l_e^{(k)}},e^{\epsilon(k)k}{l_e^{(k)}}] $, and using $\epsilon(k) \leq 1/\log^2{n}$ for all $k$, and $k \leq \log{n}$ finishes the proof.

\end{proof}


\subsubsection*{Correctness}


%
%
%

Under Assumption~\ref{assumption}, we were able to prove Lemma~\ref{lem:approx_leverage}. This, in turn, implies the correctness of our algorithm, which is to say that it generates a tree from a $\ww$-uniform distribution on trees. We now remove Assumption~\ref{assumption}, and prove the approximate correctness of our algorithm, and the first part of Theorem~\ref{thm:mainTrees}.

\mainTrees*

\begin{proof}
	Each subgraph makes at most 6 calls to $\approxS$, and there are $\log{n}$ recursive levels, so $O(n^3)$ total calls are made to $\approxS$. Setting $\delta' = \frac{\delta}{O(n^3)}$ for each call to $\approxS$, Assumption~\ref{assumption} holds with probability $(1 - \delta')^{O(n^3)} = 1-\delta$, and $\log^4{\frac{O(n^3)}{\delta}} = \tilde{O}(\log^4{1/\delta})$. Therefore, our algorithm will only fail to generate a random tree from the $\ww$-uniform distribution on trees with probability at most $\delta$
	
\end{proof}

\subsubsection{Runtime Analysis}\label{sec:runtime}

%
%
%
%
%
%
%
%
%
%
%
%
%
%
We will now analyze the runtime of the algorithm. Let $T(n)$ be the time taken by $\exact$ on input a graph $\Gtil$ with $n$ vertices and let $B(n)$ be the time taken by $\across$ on  a graph with $n$ vertices. We recall that the recursive structure then gives $T(n) = 2T(n/2) + 4B(n/2)$ and $B(n/2) = 4B(n/4)$. To compute the total runtime, we separate out the work done in the leaves of the recursion tree from the rest. Note that $\sample$ is invoked only on the leaves.

First we bound the total number of nodes of the recursion tree as a function of the depth in the tree. 

\begin{lemma}
\label{lem:nodeCount}
	Level $i$ of the recursion tree has at most $ 4^{i+1} - 2^i$ nodes, the number of vertices in the graphs at each of the nodes is at most $ n/2^i.$  
\end{lemma}

\begin{proof}

It is clear that the size of the graph at a node at depth $i$ is at most $n/2^i.$ We will bound the number of nodes by induction. There are two types of nodes in the recursion tree due to the recurrence having two kinds of branches corresponding to $T(n),B(n).$ We will call the nodes corresponding to $T(n)$ as the first type and it is clear from the recurrence relation that there are $2^i$ such nodes. Let us call the other type of nodes the second type, and it is clear that every node (both first and second type) at depth $i-1$ branches into four type two nodes. Therefore, if $a_i$ is the total number of nodes at level $i$, then $a_{i} = 4a_{i-1} + 2^{i}$. We will now prove by induction that 
$a_i \leq 4^{i+1} - 2^i$. Given $a_0 = 1$, the base case follows trivially. Suppose it is true for $i-1$, then we have 
$a_i = 4a_{i-1} + 2^{i} \leq 4( 4^{i} - 2^{i-1}) + 2^i = 4^{i+1} - 2^i,$ proving the lemma.
	
\end{proof}

Now we will compute the total work done at all levels other than the
leaves. We recall the error parameter in $\approxS$ calls is a
function of the depth in the tree:
In the case $m > n^{4/3}$, 
we define $\epsilon$ in terms of the level $i$ as
\begin{equation*}
\epsilon(i) = 2^{i/2} n^{-1/6} m^{-1/4} \log^{-2} n 
.
\end{equation*}
In the case $m \leq n^{4/3}$, 
we define $\epsilon$ in terms of the level $i$ as
\begin{equation*}
\epsilon(i) = 2^{i/2} n^{-1/2} \log^{-2} n 
.
\end{equation*}

Note that when $m > n^{4/3}$, we have $n^{-1/6} m^{-1/4} < n^{-1/2} $.
Further, the maximum value of $i$ is $\log n$ so $2^{i/2} = n^{1/2}$.
This means we always have
$\epsilon(i) \leq 2^{i/2} n^{-1/2} \log^{-2} n  \leq \log^{-2} n $.

The threshold value $t_1$ is such that $2^{2 t_1} = \frac{n^2}{m}$.


\begin{lemma}
The total work done at all levels of the recursion tree excluding the leaves is bounded by  
$\Otil(\max\setof{n^{4/3}m^{1/2},n^{2}} \log^4(1/\delta)).$
\end{lemma}

\begin{proof}
	From  Theorem~\ref{thm:schurApx} the work done in a node at
        depth $i$ is $\Otil \left( ((n/2^i)^2 +
          n2^{-i}\epsilon(i)^{-2} )\log^4(1/\delta)\right).$
        The $\log^4\left({1/\delta}\right)$ factor is left out from the remaining
        analysis for simplicity.
        By Lemma~\ref{lem:nodeCount}, the total work done at depth $i$
        is $\Otil \left( n^2  + n2^{i}\epsilon(i)^{-2}  \right).$
        Finally, bound for the total running time across all levels follows from 
\begin{align*}
\sum_{i=0}^{\log n} n^2  + 2^i \frac{n}{\epsilon(i)^2} 
& = \Otil\left( n^2 \log n + n \max\setof{n^{1/3}m^{1/2},n} \right)
.
 \end{align*}
\end{proof}

We will now analyze the total work done at the leaves of the recursion tree. We first state a lemma which gives the probability that approximate leverage score of an edge can be used to decide if the edge belongs to the tree. 
\begin{corollary}
	If $r$ is drawn uniformly randomly from $[0,1]$, then the probability that  $r \in [1 - \hat{\epsilon}\hat{l_e}\log^2{n}, 1 + \hat{\epsilon}\hat{l_e}\log^2{n}]$ is $\Otil(\hat{\epsilon}l_e)$ w.h.p.
	
\end{corollary}

\begin{proof}
	The exact probability is $2\hat{\epsilon}\hat{l_e}\log^2{n}$, and from Lemma 5.8, we know $l_e \leq 2\hat{l_e}$ w.h.p.
	
\end{proof}

We now consider the expected work done at a single leaf of the recursion tree.

\begin{lemma}
\label{lem:workedgesample}
Let $l_e$ be the leverage score of an edge $e$ in the $G$ which is obtained by updating the input graph based on the decisions made on all the edges considered before $e.$ The routine $\sample$ takes  
\[ \Otil \left(1 + l_e \max\setof{n, n^{1/3}m^{1/2}}  \right).\]
\end{lemma}
\begin{proof}
It takes $O(1)$ time to compute the  leverage score at a leaf of the
recursion tree. The routine $\sample$ successively climbs up the
recursion tree to compute the leverage score if the leverage score
estimation at the current level is not sufficient. The probability
that the outcome of $r$ is  such that we cannot make a decision at
level $i$ is $\Otil(\epsilon(i) l_e).$ 
\todolow{could elaborate}

The time required to compute the leverage score of edge $e$ in the
graph at a node at depth $i$ in the recursion tree is
$\Otil((n/2^i)^2).$

Finally, with probability 
 $\Otil(\epsilon(t_1) l_e)$
we need to compute the leverage score in the input graph and the
expected running time is $\Otil(m).$
Therefore, when $m > n^{4/3}$ and $\eps(i) = 2^{i/2} n^{-1/6} m^{-1/4}
\log^{-2} n $,
the total expected running time is 
\begin{align*}
\Otil \left(
1
+
m \epsilon(t_1) l_e + l_e \sum_{i=t_1}^{i=\log n} \epsilon(i)
  \frac{n^2}{4^i}
\right)
&=
\Otil \left(
1
+
l_e
m
n^{-1/6} m^{-1/4}
n^{1/2} m^{-1/4}
\right)
\\
&= 
\Otil \left(1 + l_e n^{1/3}m^{1/2}  \right)
.
\end{align*}
When $m \leq n^{4/3}$ and $\eps(i) = 2^{i/2} n^{-1/2}
\log^{-2} n $,
the total expected running time is 
\begin{align*}
\Otil \left(
1
+
m \epsilon(t_1) l_e + l_e \sum_{i=t_1}^{i=\log n} \epsilon(i)
  \frac{n^2}{4^i}
\right)
&=
\Otil \left(
1
+
l_e
m
n^{-1/2}
n^{1/2} m^{-1/4}
\right)
\\
&= 
\Otil \left(1 + l_e m^{3/4}  \right)
\\
&= 
\Otil \left(1 + l_e n  \right)
.
\end{align*}
Note that $n^{1/3}m^{1/2} \geq n$ if and only if $m \geq n^{4/3}$, so
we can summarize this as the expected running time being bounded by 
$\Otil \left(1 + l_e \max\setof{n, n^{1/3}m^{1/2}}  \right)$.


%
%
%
	
\end{proof}

We now want to give the runtime cost over all edges. Let us label the edges $ e_1,...,e_m $ in the order in which the decisions are made on them. In the following, when we talk about leverage score $l_{e_i}$ of an edge $e_i,$ we mean the leverage score of the edge $e_i$ in the graph obtained by updating $G$ based on the decisions made on $e_1,...,e_{i-1}.$ 

	\begin{lemma}
		Let $e_i$ be the first edge sampled to be in the tree, and  $X = l_{e_1} + l_{e_2} + l_{e_3} + ... + l_{e_i}$ be a random variable.  Then, $$ \Pr (X > C) \leq e^{-C}.$$ 
		
	\end{lemma}
	
	\begin{proof}
		Let $p_j = l_{e_j}$, we have $0 \leq p_j \leq 1$. If $\sum_j p_j \geq C$, then the probability that the edges $e_1,...,e_{i-1}$ is deleted is 
$$\prod_{j=1}^i (1 - p_j) \leq \left( 1 - \frac{C}{i} \right)^i \leq e^{-C}.$$ 
	\end{proof}

We thus have $E(X) = O(1)$, and also, with probability at least $ 1 - 1/\text{poly}(n)$ we have $X = O(\log n).$ 
Applying this iteratively until $n-1$ edges are sampled to be in the tree, we have that the expected sum of conditional leverage scores is $O(n)$, and is $O(n \log n)$ with probability $ 1 - 1/\text{poly}(n).$

\begin{corollary}
	The total expected work done over all the leaves of the
        recursion tree is $\Otil(\max\setof{n^{4/3}m^{1/2},n^{2}}).$
\end{corollary}

\begin{proof}
This immediately follows from Lemma~\ref{lem:workedgesample} by plugging in $\sum_e l_e = O(n \log n),$ which holds with probability at least $1- 1/\text{poly}(n),$ and observing that the work done at the leaves is $\text{poly}(n)$ in the worst case.

\end{proof}


\section{Schur Complement Approximation}
\label{sec:schurapx}

In this section, we give an algorithm for spectral approximation of
the Schur complement of a Laplacian matrix.
Our approach closely
 follows that in \cite{KyngS16}, with the main 
distinction being: We show that if their algorithm is used to
eliminate only part of the original set of vertices, then the
remaining matrix is a good 	spectral approximation of the Schur
complement. We also combined their algorithm with additional
leverage score estimation and sparsification to produce a sparser
output.

\subsection{Preliminaries}
This subsection is mostly replicated from \cite{KyngS16} for the sake of
completeness.
We start by introducing Cholesky factorizations and Schur complements.
Conventionally, these matrix operations are understood in terms of
factorizations into lower triangular matrices.
We will instead present an an equivalent view
where the Schur complement is obtained by iteratively subtracting rank
one terms from a matrix.
Let $\LL$ be the Laplacian of a connected graph.
Let $\LL(:,i)$ denote the $i^{\text{th}}$ column of $\LL$.
\[ \SS^{(1)} \defeq \LL - \frac{1}{\LL(1,1)}\LL(:,1) \LL(:,1)^{\top},\]
is called the \emph{Schur complement} of $\LL$ with respect to vertex
$1$. $\SS^{(1)}$ are identically 0, and thus this is effectively a system in
the remaining $n-1$ indices.

More generally, we can compute the Schur complement w.r.t. any single
vertex (row and column index) of $\LL$. 
Suppose we want the Schur complement w.r.t. vertex $v_{1}$.
Letting
$\alpha_{1} \defeq \LL(v_{1},v_{1}), \cc_{1} \defeq
\frac{1}{\alpha_{1}} \LL(:,v_{1}),$ we have $\LL = \SS^{(1)} + \alpha_{1}
\cc_{1} \cc_{1}^{\top}.$

We can also perform a sequence of eliminations, where in the $i^{th}$
step, we select a vertex
$v_{i} \in V\setminus \setof{v_{1}, \ldots, v_{i-1}}$ and eliminate
the vertex $v_{i}.$ We define
\begin{align*}
  \alpha_{i} &= \SS^{(i-1)}(v_{i},v_{i}) \\
  \cc_{i} &= \frac{1}{\alpha_{i}} \SS^{(i-1)}(:,v_{i}) \\
  \SS^{(i)} &= \SS^{(i-1)} - \alpha_{i} \cc_{i}\cc_{i}^{\top}.
\end{align*}
If at some step $i$, 
$\SS^{(i-1)}(v_{i},v_{i}) = 0$,
then we define $\alpha_{i} = 0$,
and $\cc_{i} = 0$. However, when the original matrix is the Laplacian
of a connected graph, it can be shown that every choice of $v_{i}$
gives a non-zero $\alpha_{i}$, and that the resulting matrix $\SS^{(i)}$  is always
the Laplacian of a connected graph.

While it does not follow immediately from the above, it is a
well-known fact that the Schur complement $\SS^{(i)}$ w.r.t. a sequence
of variables $v_{1}, \ldots, v_{i}$ does not depend on the order in
which the vertices are eliminated (but the $\cc_{i}$ and $\alpha_{i}$
do depend on the order).
Consequently it makes sense to define $\SS^{(i)}$ as the Schur complement
w.r.t. elimination of the \emph{set} of vertices $\setof{v_{1},
  \ldots, v_{i}}$ (see Fact~\ref{fac:elimOrderEquiv}).

Suppose we eliminate a sequence of vertices $v_{1}, \ldots, v_{j}$
Let $\matlow$ be the 
$n \times j$ matrix
with $\cc_{i}$ as its $i^{\textrm{th}}$ column, 
and $\calD$ be the $n \times j$ diagonal matrix 
$\calD(i,i) = \alpha_{i}$, then
\[
\LL = \SS^{(j)} + \sum_{i=1}^{j} \alpha_{i} \cc_{i} \cc_{i}^{\top} 
=\SS^{(j)} + \matlow \calD \matlow^{\top}.
\]
This decomposition is known a partial Cholesky factorization.
Let us write $F = \setof{v_{1},\ldots,v_{j}}$, and $C = V-F$.
We can then write
$\matlow = \begin{pmatrix} \matlow_{FF} \\ \matlow_{CF} \end{pmatrix} $.
If we abuse notation and also identify $\SS^{(j)}$ if the matrix
restricted to its non-zero support $C$,
then we can also write
\begin{align}
  \LL =
  \begin{pmatrix}
    \matlow_{FF} & \matzero \\
    \matlow_{CF} & \II_{CC}
  \end{pmatrix}
  \begin{pmatrix}
  \calD & \matzero \\
  \matzero & \SS^{(j)}
  \end{pmatrix}
  \begin{pmatrix}
    \matlow_{FF} & \matzero \\
    \matlow_{CF} & \II_{CC}
  \end{pmatrix}^{\top}
\end{align}

\paragraph{Clique Structure of the Schur Complement.}
Given a Laplacian $\LL$,
let $\vstar{\LL}{v} \in \rea^{n \times n}$ denote the Laplacian corresponding to the 
edges incident on vertex $v$, i.e.
\begin{equation}
\label{eq:starDef}
\vstar{\LL}{v} \defeq \sum_{e \in E : e \ni v} w(e) \bb_{e} \bb_{e}^{\top}
.
\end{equation}
For example, we denote the first column of $\LL$ by
$
\begin{pmatrix}
  d \\
  -\aa 
\end{pmatrix}
,$
then
$
\vstar{\LL}{1} =
\begin{bmatrix}
d & -\aa^{\top} \\
-\aa & \diagop(\aa)
\end{bmatrix}
.
$
We can write the Schur complement 
$\SS^{(1)}$ w.r.t. a vertex $v_{1}$ as
$\SS^{(1)} = \LL-\vstar{\LL}{v_{1}} +\vstar{\LL}{v_{1}}  - \frac{1}{\LL(v_{1},v_{1})}\LL(:,v_{1})\LL(:,v_{1})^{\top}.$
It is immediate that $L-\vstar{L}{v_{1}} $ is a Laplacian matrix,
since $\LL-\vstar{\LL}{v_{1}}  = \sum_{e \in E : e \not\ni v_{1}} w(e) \bb_{e}
\bb_{e}^{\top}$.
A more surprising (but well-known) fact is that 
\begin{align}
\label{eq:Cv-def}
C_{v_{1}}(\LL) \defeq \vstar{\LL}{v_{1}}  - \frac{1}{\LL(v_{1}, v_{1})}\LL(:, v_{1}) \LL(:, v_{1})^{\top}
\end{align}
is also a Laplacian, and its edges form a clique on the neighbors of $v_{1}$.
It suffices to show it for $v_{1} = 1.$ We write $i \sim j$ to denote $(i,j) \in E.$
Then
\begin{equation*}
\label{eq:cliquestructure}
C_{1}(\LL)
=
\LL_{1} - \frac{1}{\LL(1,1)}\LL(:,1) \LL(:,1)^{\top}
=
\begin{bmatrix}
\matzero & \veczero^{\top} \\
\veczero & \diagop(\aa)- \frac{\aa\aa^{\top} }{d}
\end{bmatrix}
= 
\sum_{i \sim 1} \sum_{j \sim 1} 
\frac{w(1,i) w(1,j)} 
{d}
\bb_{(i,j)} \bb_{(i,j)}^{\top}
.
\end{equation*}
Thus
$\SS^{(1)}$ is a Laplacian since it is a sum of two Laplacians.
By induction, for all $k,$ $\SS^{(k)}$ is a Laplacian.
Thus: 
\begin{fact}
  The Schur complement of a Laplacian w.r.t. vertices $v_{1}, \ldots,
  v_{k}$ is a Laplacian.
\end{fact}
\subsection{Further Properties of the Schur Complement and Other Factorizations}

Consider a general PSD matrix of the form
\begin{align}
  \MM =
  \begin{pmatrix}
    \AA & \matzero \\
    \BB & \II
  \end{pmatrix}
  \begin{pmatrix}
  \RR & \matzero \\
  \matzero & \TT
  \end{pmatrix}
  \begin{pmatrix}
    \AA & \matzero \\
    \BB & \II
  \end{pmatrix}^{\top}
\end{align}
where $\AA$ is invertible and $\II$ is the identity matrix on a subset
of the indices of $\MM$.
It is easy to show the following well-known fact:
\begin{fact}
\label{fac:pinvproduct}
  Suppose $\XX$ is a non-singular matrix and $\AA$ is a symmetric
  matrix, and $\PP$ is the orthogonal projection to the complement of
  the null space of $\XX \AA \XX^{\top}$.
Then $(\XX \AA \XX^{\top})^{+} = P\XX^{-1} \AA^{+} \XX^{-\top}P$.
\end{fact}
Based on Fact~\ref{fac:pinvproduct} for vectors
orthogonal to null space of $\MM$ we have
\begin{align*}
  \xx^{\top} \MM^+ \xx
 &=
\xx^{\top} 
  \begin{pmatrix}
    \AA & \matzero \\
    \BB & \II
  \end{pmatrix}^{-\top}
  \begin{pmatrix}
  \RR^{+} & \matzero \\
  \matzero & \TT^{+} 
  \end{pmatrix}
  \begin{pmatrix}
    \AA & \matzero \\
    \BB & \II
  \end{pmatrix}^{-1}
\xx
\end{align*}
Recall the general formula for blockwise inversion: 
\begin{align*}
\begin{pmatrix}
 \AA & \CC \\
 \BB & \DD 
\end{pmatrix}^{-1}
=
\begin{pmatrix}
\AA^{-1} +\AA^{-1} \CC (\DD-\BB\AA^{-1}\CC)^{-1} \BB \AA^{-1} 
&
-\AA^{-1} \CC (\DD-\BB\AA^{-1}\CC)^{-1} 
\\
- (\DD-\BB\AA^{-1}\CC)^{-1}\BB \AA^{-1}
&
 (\DD-\BB\AA^{-1}\CC)^{-1}   
\end{pmatrix} 
\end{align*}
Thus by applying the formula for blockwise inversion and simplifying,
we get
\begin{align*}
  \begin{pmatrix}
    \AA & \matzero \\
    \BB & \II
  \end{pmatrix}^{-1}
=
\begin{pmatrix}
\AA^{-1}
&
\matzero
\\
-\BB \AA^{-1}
&
\II
\end{pmatrix} 
\end{align*}
So
\begin{align*}
  \xx^{\top} \MM^+ \xx
 &=
\xx^{\top} 
\begin{pmatrix}
\AA^{-1}
&
\matzero
\\
-\BB \AA^{-1}
&
\II
\end{pmatrix}^{\top}
  \begin{pmatrix}
  \RR^{+} & \matzero \\
  \matzero & \TT^{+} 
  \end{pmatrix}
\begin{pmatrix}
\AA^{-1}
&
\matzero
\\
-\BB \AA^{-1}
&
\II
\end{pmatrix} 
\xx
.
\end{align*}
Suppose
$\xx =
\begin{pmatrix}
\matzero \\
\yy
\end{pmatrix}
$, and again $\xx$ is orthogonal to the null space of $\MM$.
Then

\begin{align}
\label{eq:pinvQFschur}
  \xx^{\top} \MM^+ \xx
 &=
\begin{pmatrix}
\matzero \\
\yy
\end{pmatrix}^{\top} 
\begin{pmatrix}
\AA^{-1}
&
\matzero
\\
-\BB \AA^{-1}
&
\II
\end{pmatrix}^{\top}
  \begin{pmatrix}
  \RR^{+} & \matzero \\
  \matzero & \TT^{+} 
  \end{pmatrix}
\begin{pmatrix}
\AA^{-1}
&
\matzero
\\
-\BB \AA^{-1}
&
\II
\end{pmatrix} 
\begin{pmatrix}
\matzero \\
\yy
\end{pmatrix}
\\
& 
\nonumber
=
\yy^{\top} \TT^{+}  \yy
.
\end{align}
Consider a partial Cholesky decomposition of a connected Laplacian
$\LL$ w.r.t. elimination of the sequence of verties $v_{1}, \ldots, v_{j}$, 
where we write $F = \setof{v_{1}, \ldots, v_{j}}$ and $C = V-F$.
Recall that the resulting
Schur complement $\SS$ is another Laplacian.
\begin{align}
  \LL =
  \begin{pmatrix}
    \matlow_{FF} & \matzero \\
    \matlow_{CF} & \II_{CC}
  \end{pmatrix}
  \begin{pmatrix}
  \calD & \matzero \\
  \matzero & \SS
  \end{pmatrix}
  \begin{pmatrix}
    \matlow_{FF} & \matzero \\
    \matlow_{CF} & \II_{CC}
  \end{pmatrix}^{\top}
\end{align}
Note that as $\SS$ is a connected Laplacian on a subset of the
vertices of $\LL$.
\begin{fact}
\label{fac:elimOrderEquiv}
The Schur complement of a connected Laplacian $\LL$ w.r.t. to a sequence  of
vertices $v_{1}, \ldots, v_{j}$ does not depend on the order of
elimination of these vertices. Let $C = V-\setof{v_{1}, \ldots,
  v_{j}}$, then the Schur complement is equivalent to the Schur
complement $\textsc{Schur}(G,
C)$ as stated in Definition~\ref{def:schurcomp}.
\end{fact}
\begin{proof}
Suppose we use two orderings on the variables $v_{1}, \ldots, v_{j}$
to produce factorizations
\begin{align}
  \LL =
  \begin{pmatrix}
    \matlow_{FF} & \matzero \\
    \matlow_{CF} & \II_{CC}
  \end{pmatrix}
  \begin{pmatrix}
  \calD & \matzero \\
  \matzero & \SS
  \end{pmatrix}
  \begin{pmatrix}
    \matlow_{FF} & \matzero \\
    \matlow_{CF} & \II_{CC}
  \end{pmatrix}^{\top}
\label{eq:lapfactor1}
\end{align}
and
\begin{align}
  \LL =
  \begin{pmatrix}
    \matlowhat_{FF} & \matzero \\
    \matlowhat_{CF} & \II_{CC}
  \end{pmatrix}
  \begin{pmatrix}
  \calDhat & \matzero \\
  \matzero & \SShat
  \end{pmatrix}
  \begin{pmatrix}
    \matlowhat_{FF} & \matzero \\
    \matlowhat_{CF} & \II_{CC}
  \end{pmatrix}^{\top}
\label{eq:lapfactor2}
\end{align}
where we use $F = \setof{v_{1}, \ldots, v_{j}}$ and $C = V-F$.
Furthermore, both $\SS$ and $\SShat$ have a null space that is exactly
the span of $\vecone_{C}$. 
We can see this in two steps: Firstly, both
are Laplacian matrices, so their null spaces must include the span of
$\vecone_{C}$. 
Secondly, from the product forms in
Equations~\eqref{eq:lapfactor1} and~\eqref{eq:lapfactor2}, if either
had null space of rank strictly larger than 1, then the rank of $\LL$ would be strictly less than $n-1$, which is false.
Consider $\xx =
\begin{pmatrix}
\matzero \\
\yy
\end{pmatrix}
$, where $\yy$ is orthogonal to $\vecone_{C}$ and hence $\xx$ is
orthogonal to $\vecone$.
By Equation~\eqref{eq:pinvQFschur}, $\xx^{\top}\LL^{+}\xx =
\yy^{\top} \SS^{+}\yy = \yy^{\top} \SShat^{+}\yy $.
Which also implies $\yy^{\top} \SS \yy = \yy^{\top} \SShat \yy $
for all vectors $\yy$ orthogonal to  $\vecone_{C}$.
As $\SS$ and $\SShat$ have the same null space, we then conclude $\SS
= \SShat$.

We can apply the same reasoning to the factorization
\begin{align}
  \LL =
  \begin{pmatrix}
    \AA & \BB \\
    \BB^{\top} & \CC
  \end{pmatrix}
=
  \begin{pmatrix}
    \II & \matzero \\
    \BB^{\top} \AA^{-1} & \II
  \end{pmatrix}
  \begin{pmatrix}
  \AA & \matzero \\
  \matzero & \CC-\BB^{\top} \AA^{-1} \BB
  \end{pmatrix}
  \begin{pmatrix}
    \II & \matzero \\
    \BB^{\top} \AA^{-1} & \II
  \end{pmatrix}
\label{eq:lapfactor3}
\end{align}
and conclude $\SS = \CC-\BB^{\top} \AA^{-1} \BB$, so 
Definition~\ref{def:schurcomp} of the Schur complement is equivalent to the obtained by
a sequence of eliminations.

\end{proof}
From the above proof, it we also immediately get the following fact:
\begin{fact}
\label{fac:pinvandschur}
Consider a connected Laplacian $\LL$ and a subset  $F\subseteq V$ of
its vertices, and let $\SS$ be
the Schur complement of $\LL$ w.r.t. elimination of $F$.
Let $C = V-F$. 
Suppose $\xx = \begin{pmatrix} \xx_{F} \\ \xx_{C} \end{pmatrix}$
is a vector orthogonal to the null space of $\LL$,
and $\xx_{F} =\matzero$.

Then $\xx^{\top}\LL^{+}\xx = \xx_{C}^{\top}\SS^{+}\xx_{C}$.
\end{fact}
\subsection{Spectral Aproximation of the Schur Complement}

Theorem~\ref{thm:schurApx}, stated below, characterizes the
performance of our algorithm~{\scElim}.
This algorithm computes a spectral approximation of the Schur
complement of a Laplacian w.r.t elimination of a set of vertices
$F =V-C$.
The algorithm relies on three procedures:
\begin{itemize}
\item
{\levEst}, which
computes approximate leverage scores of all edges in a graph;
 The guarantees of {\levEst} are given in Lemma~\ref{lem:levest}.
\item
{\sparsify}, which sparsifies a graph.
 {\sparsify} is characterized in Lemma~\ref{lem:levest}.
\item
{\textsc{CliqueSample}} which returns a sparse Laplacian matrix
  approximating a clique created by elimination (see \cite{KyngS16},
  Algorithm 2).
\end{itemize}
The pseudocode for {\scElim} is given in Algorithm~\ref{alg:schurapx}.

\schurApx*

Lemma~\ref{lem:levest} stated below follows immediately from using
the Laplacian solver of \cite{KoutisMP11} in the effective resistance
estimation procedure of \cite{SpielmanS11}.

\begin{lemma}
\label{lem:levest}
 Given a connected undirected multi-graph
  $G =(V,E)$, with positive edges weights 
  $w : E \to \rea_{+}$, and associated Laplacian $\LL$,
  and a scalar $0 < \delta < 1$
  the algorithm $\levEst(\LL,\delta)$ returns estimates $\tauEst_e$ for all the edges such that
  with probability $\geq 1-\delta$
	\begin{enumerate}
		\item For each edge $e$, we have
		$\tau_{e} \leq \widehat{\tau}_e \leq 1$ where
		$\tau_{e}$ is the true leverage score of $e$ in $G$.
        \item $ \sum_{e} \tauEst_e
			\leq 2 n.$
	\end{enumerate}
The algorithm runs in time $O(m \log^{2} (n/\delta)
\operatorname{polyloglog}(n) )$.
\end{lemma}

Lemma~\ref{lem:levest} stated below follows immediately from using
the Laplacian solver of \cite{KoutisMP11} in the sparsification
routine of \cite{SpielmanS11}.

\begin{lemma}
\label{lem:sparsify}
Given a connected undirected multi-graph
  $G =(V,E)$, with positive edges weights 
  $w : E \to \rea_{+}$, and associated Laplacian $\LL$,
and scalars $0<\eps\leq1/2$, $0< \delta < 1$,
$\sparsify(\LL,\eps,\delta)$ returns a Laplacian $\LLtil$ s.t.
with probability $\geq 1 - \delta$ it holds that
 $\LLtil
\approx_{\eps} \LL$ and $\LLtil$ has $O(n \eps^{-2} \log(n/\delta) )$ edges.
The algorithm runs in time  $O( m \log^{2}(n/\delta) \operatorname{polyloglog}(n)
+ n \eps^{-2} \log(n/\delta)  )$.
\end{lemma}

\begin{algorithm}
	\caption{$\scElim(L,C,\epsilon,\delta)$}
	\label{alg:schurapx}
        Call $\levEst(\LL,\delta/3)$ to compute leverage score estimates
        $\tauEst_e$ for every edge $e$\;
        \For{every edge $e$}
        {
         $\SStil^{(0)} \leftarrow \LLtil$ with multi-edges split into 
$\rho_{e} =
\ceil{\tauEst_{e} \cdot 12
  \left(\frac{\eps}{2}\right)^{-2} \ln^{2}(3 n/\delta)}$
 copies with
$1/\rho_{e} $ of the original weight\;
        }
  Let $F = V-C$\;
  Label the vertices in $F$ by $\setof{1, \ldots, \sizeof{F}}$
  and the remaining vertices $C$ by $\setof{\sizeof{F}+1,
    \ldots, n}$ \;
    Let $\pi$ be a uniformly random permutation on  $\setof{1, \ldots,
      \sizeof{F}}$\; 
    \For{$i = 1$ to $\sizeof{F}$}
    {
          $\CCtil_{i}\leftarrow
          \csamp(\SStil^{(i-1)},\pi(i))$\;
          $\SStil^{(i)} \leftarrow \SStil^{(i-1)} -
          \vstar{\SStil^{(i-1)}}{\pi(i)} + \CCtil_{i}$\;
    }
 $\SStil \gets \sparsify(\SStil^{(\sizeof{F})},\eps, \delta/3)$\;
 \KwRet {$\SStil$\;}
\end{algorithm}

Our proof of Theorem~\ref{thm:schurApx} relies on the following lemma which
provides a similar, but seemingly weaker guarantee
about the output of the algorithm {\partialChol}. Its pseudo-code is
given in Figure~\ref{alg:partialChol}.
\begin{lemma}
\label{lem:partialChol}
  Given a connected undirected multi-graph
  $G =(V,E)$, with positive edges weights 
  $w : E \to \rea_{+}$, and associated Laplacian $\LL$,
  a set vertices $C \subset V$,
  and scalars $0< \delta <1$, $0<\eps\leq1/2$,
  the algorithm\\
  $\partialChol(\LL,C,\eps)$
  returns a decomposition
  $(\matlowtil,\calDtil,\SStil)$.
  With probability $\geq 1-\delta$, the following statements all hold:
  \begin{align}
    \label{eq:lapErrorBounds}
   \LL
   \approx_{\eps}
     \LLtil
  \end{align}
where $F = V-C$ and
\[
\LLtil =
  \begin{pmatrix}
    \matlowtil_{FF} \\
    \matlowtil_{CF}
  \end{pmatrix}
  \calDtil
  \begin{pmatrix}
    \matlowtil_{FF} \\
    \matlowtil_{CF}
  \end{pmatrix}^{\top}
+ 
\begin{pmatrix}
\matzero_{FF} & \matzero_{FC} \\
\matzero_{CF} & \SStil
\end{pmatrix}
.
\]
Here $\SStil$ is a Laplacian matrix whose edges are supported on $C$.
Let $k = \sizeof{C} = n-\sizeof{F}$.
The total number of non-zero entries $\SStil$ is $O(k
\eps^{-2}\log(n/\delta))$.
$\matlowtil_{FF}$ is an invertible matrix.
The total number of non-zero entries in $\matlowtil_{FF}$ and
$\matlowtil_{FC}$ is $O(m + n \eps^{-2}\log n \log(n/\delta)^{2}  )$.
The total running time is bounded by
$O((m\log n \log^{2}(n/\delta)+
n \eps^{-2} \log n \log^{4} (n/\delta))
\operatorname{polyloglog}(n))$.
\end{lemma}
\begin{algorithm}
	\caption{$\partialChol(\LL,C,\epsilon,\delta)$}
	\label{alg:partialChol}
        Call $\levEst(\LL ,\delta/3)$ to compute leverage score estimates
        $\tauEst_e$ for every edge $e$\;
        \For{every edge $e$}
        {
         $\SStil^{(0)} \leftarrow \LLtil$ with multi-edges split into $\rho_{e} =
\ceil{\tauEst_{e} \cdot 12
  \left(\frac{\eps}{2}\right)^{-2} \ln^{2}(3 n/\delta)}$ copies with
$1/\rho_{e} $ of the original weight\;
        }
        Let $F = V-C$\;
        Define the diagonal matrix $\calDtil \gets \matzero_{\sizeof{F}
    \times \sizeof{F}}$\;
   Label the vertices in $F$ by $\setof{1, \ldots, \sizeof{F}}$
  and the remaining vertices $C$ by $\setof{\sizeof{F}+1,
    \ldots, n}$ \;
    Let $\pi$ be a uniformly random permutation on  $\setof{1, \ldots,
      \sizeof{F}}$\; 
    \For{$i = 1$ to $\sizeof{F}$}
    {
          $\calDtil(i,i) \leftarrow (\pi(i),\pi(i))$
          entry of $\SStil^{(i-1)}$\;
          $\cctil_{i} \leftarrow \pi(i)^{\text{th}} \text{ column of }
          \SStil^{(i-1)}$ divided by $\calDtil(i,i)$ 
   if $\calDtil(i,i) \neq 0,$ or zero otherwise\;
          $\CCtil_{i}\leftarrow
          \csamp(\SStil^{(i-1)},\pi(i))$\;
          $\SStil^{(i)} \leftarrow \SStil^{(i-1)} -
          \vstar{\SStil^{(i-1)}}{\pi(i)} + \CCtil_{i}$\;
    }
$\matlowtil  \gets 
   \begin{pmatrix}
    \cc_{1} & \cc_{2} & \ldots & \cc_{\sizeof{F}}
  \end{pmatrix}
 $\;
 $\SStil \gets \sparsify(\SStil^{(\sizeof{F})},\eps, \delta/3)$\;
\KwRet{
  $(
\matlowtil, \calDtil, \SStil) $\;}
\end{algorithm}

\pagebreak
\begin{proof}(of Theorem~\ref{thm:schurApx})
Note, 
given the elimination ordering $\pi(v_{1}), \ldots,
\pi(v_{\sizeof{F}})$
we can do a partial Cholesky factorization of $\LL$ as
\begin{align}
  \LL =
  \begin{pmatrix}
    \matlow_{FF} & \matzero \\
    \matlow_{CF} & \II_{CC}
  \end{pmatrix}
  \begin{pmatrix}
  \calD & \matzero \\
  \matzero & \SS
  \end{pmatrix}
  \begin{pmatrix}
    \matlow_{FF} & \matzero \\
    \matlow_{CF} & \II_{CC}
  \end{pmatrix}^{\top}
\label{eq:lapfactor}
\end{align}
where $\SS$ is the Schur complement of $\LL$ w.r.t. $F$.

We note that {\scElim} and {\partialChol} perform exactly the same
computations, with the exception that $\partialChol$ records the
values $\matlowtil$ and $\calDtil$.
This means we can establish a simple coupling between the algorithms
by considering them executing based on the same source of randomness:
They must then return the same matrix $\SStil$.
Thus, if we can show for the matrix $\SStil$ returned by
$\partialChol$ that $\SStil \approx_{\eps} \SS$, then the same must be
true for the $\SStil$ returned by $\scElim$.

We can write the matrix $\LLtil$ constructed from the output of
$\partialChol$ as 
\begin{align}
  \LLtil =
  \begin{pmatrix}
    \matlowtil_{FF} & \matzero \\
    \matlowtil_{CF} & \II_{CC}
  \end{pmatrix}
  \begin{pmatrix}
  \calDtil & \matzero \\
 \matzero & \SStil
  \end{pmatrix}
  \begin{pmatrix}
    \matlowtil_{FF} & \matzero \\
    \matlowtil_{CF} & \II_{CC}
  \end{pmatrix}^{\top}
\label{eq:laptilfactor}
\end{align}
We now suppose that $\partialChol$ succeeds and returns $\LLtil
\approx_{\eps} \LL$. These two matrices must have the same null
space, namely the span of $\vecone$.
Consider $\xx =
\begin{pmatrix}
\matzero \\
\yy
\end{pmatrix}
$, where $\yy$ is orthogonal to $\vecone_{C}$ and hence $\xx$ is
orthogonal to $\vecone$.
By Equation~\eqref{eq:pinvQFschur}, $\xx^{\top}\LLtil^{+}\xx =
\yy^{\top} \SStil^{+} \yy$, and $\xx^{\top}\LL^{+}\xx =
\yy^{\top} \SS^{+} \yy$. 
$\LLtil \approx_{\eps} \LL$ implies $\LLtil^{+} \approx_{\eps}
\LL^{+}$, and so
\begin{align}
\exp(-\eps) \yy^{\top} \SS^{+} \yy \leq \yy^{\top} \SStil^{+} \yy \leq \exp(\eps)
  \yy^{\top} \SStil^{+} \yy.
\label{eq:schurApxOrthogToOne}
\end{align}
Furthermore, both $\SS$ and $\SStil$ have a null space that is exactly
the span of $\vecone_{C}$. 
We can see this in two steps: Firstly, both
are Laplacian matrices, so their null spaces must include the span of
$\vecone_{C}$. 
Secondly, from the product forms in
Equations~\eqref{eq:lapfactor} and~\eqref{eq:laptilfactor}, if either
had null space of rank strictly larger than 1, then the rank of $\LL$
or $\LLtil$ would be strictly less than 1, which is false.
So by
contradiction, both $\SS$ and $\SStil$ have a null space that is exactly
the span of $\vecone_{C}$.
From this and
Equation~\eqref{eq:schurApxOrthogToOne},
which holds for all $\yy$
orthogonal to $\vecone_{C}$, we conclude $\SStil^{+}
\approx_{\eps} \SS^{+}$.
This in turn implies $\SStil
\approx_{\eps} \SS$.

The guarantees of success probability, running time and sparsity of
$\SStil$ for $\scElim$ now follow from the guarantees for
$\partialChol$ given in Lemma~\ref{lem:partialChol}. 
\end{proof}

\subsection{Properties of Approximate Partal Cholesky Factorization}

In this subsection, we prove Lemma~\ref{lem:partialChol}, which
describes the main guarantee of algorithm~{\partialChol} (Algorithm~\ref{alg:partialChol}).
The algorithm~{\partialChol} is obtained from the algorithm
\textsc{SparseCholesky} given in \cite{KyngS16} by making four small
modifications:
\begin{enumerate}
\item Instead of splitting every original edge into the same number of
  smaller copies, edges are split into smaller copies based on
  estimates of their leverage score.
\item {\partialChol}  only eliminates a subset of the vertices.
This restricts the choices random vertices available to eliminate,
which increases the variance of the algorithm per round of
elimination.
But it also decreases the number of rounds of elimination, which
decreases the total variance accumulated over all rounds of
elimination.
\item To make the matrix of eliminated columns lower-triangular,
\textsc{SparseCholesky} permutes the rows (see \textsc{SparseCholesky}
algorithm Line 10).
{\partialChol} does not need the matrix of eliminated columns to be lower-triangular, so we
do not apply this permutation. 
\item
Algorithm~{\partialChol} outputs a matrix decomposition
  $(\matlowtil,\calDtil,\SStil)$. 
The matrix composition $(\matlowtil,\calDtil,\SStil^{(\sizeof{F})})$
corresponds to an intermediate result computed by
\textsc{SparseCholesky}, but rather than directly outputting this
result, {\partialChol} first applies $\sparsify$ to
$\SStil^{(\sizeof{F})}$ to compute the sparser approximation $\SStil$.
\end{enumerate}

We now sketch a proof of Lemma~\ref{lem:partialChol}, by addressing
how the proof of correctness for \textsc{SparseCholesky}  in \cite{KyngS16} can be adapted to accommodate the changes listed above.

One can prove Lemma~\ref{lem:partialChol} using exactly the Martingale
framework developed in \cite{KyngS16} and applied in their proof of their Theorem~3.1.

\begin{proof}(Sketch of Lemma~\ref{lem:partialChol})
We describe how to address the changes listed above:
\begin{enumerate}
\item
In \cite{KyngS16} it is proven that their algorithm
\textsc{SparseCholesky} succeeds in producing a sparse approximate
Cholesky factorization with probability $1-\delta$, when 
started with a multi-graph where all multi-edges have leverage
score at most $\frac{1}{12\eps^{-2} \ln^{2} (n/\delta) }$.
The \textsc{SparseCholesky} algorithm achieves this bound on leverage
scores by using that original edges have leverage score at most $1$,
and then splitting all original edges into $\rho = \ceil{12
  \eps^{-2} \ln^{2}(n/\delta) }$ copies with weight $1/\rho$ of the original.
This bounds the norms of the multi-edges as desired, while ensuring a
total of at most $\rho m$ multi-edges.
The running time and final number of non-zeros in the output of \textsc{SparseCholesky} is equal to $O(\log n)$ times
the number of multi-edges in the graph after splitting edges, so
it is bounded by $O(\rho m \log n) = O(  \eps^{-2} \log n \log^{2}(n/\delta) )  $.

The algorithm {\partialChol} first computes leverage score estimates 
by the call to $\levEst$, which by Lemma~\ref{lem:levest}
succeeds with probability
$1-\delta/3$ and returns leverage score estimates $\tauEst_{e}$ 
that upper bound the true leverage scores $\tau_{e}$, while ensuring
$\sum_{e} \tauEst_{e} \leq 2n$.
It then splits each edge $e$ into $\rho_{e} =
\ceil{12 \tauEst_{e} 
  \left(\frac{\eps}{2}\right)^{-2} \ln^{2}(3 n/\delta) }$ copies with 
weight $1/\rho_{e}$ of the original.
This ensures a bound on the leverage score of each multi-edge of 
$\frac{\tau_{e} }{12 \tauEst_{e} 
  \left(\frac{\eps}{2}\right)^{-2} \ln^{2}(3 n/\delta)}
\leq
\frac{1}{ 12
  \left(\frac{\eps}{2}\right)^{-2} \ln^{2}(3 n/\delta)}$.
This is the same as the leverage score bound achieved by \textsc{SparseCholesky},
except with $\delta$ replaced by $\delta/3$ and $\eps$ replaced by $\eps/2$.
Thus the elimination procedure should succeed with probability
$1-\delta/3$, and achieve
\[
\LL \approx_{\eps/2}
  \begin{pmatrix}
    \matlowtil_{FF} \\
    \matlowtil_{CF}
  \end{pmatrix}
  \calDtil
  \begin{pmatrix}
    \matlowtil_{FF} \\
    \matlowtil_{CF}
  \end{pmatrix}^{\top}
+ 
\begin{pmatrix}
\matzero_{FF} & \matzero_{FC} \\
\matzero_{CF} & \SStil^{(\sizeof{F})}
\end{pmatrix}
.
\]
The total number of multi-edges created by the intial splitting in
{\partialChol} will be $O(\ceil{12 \tauEst_{e} 
  \left(\frac{\eps}{2}\right)^{-2} \ln^{2}(3 n/\delta) }) = O(\sum_{e}
1 + 12 \tauEst_{e} 
  \left(\frac{\eps}{2}\right)^{-2} \ln^{2}(3 n/\delta) ) = O(m + n
  \left(\frac{\eps}{2}\right)^{-2} \ln^{2}(3 n/\delta) )$.
The final number of non-zeros in
$\SStil^{(\sizeof{F})}$ and the time required for the approximate
eliminations will both be
upper bounded by $O(\log n)$ times the
initial number of multi-edges so upper bounded by $O(m\log n +
n \eps^{-2} \log n \log^{2} (n/\delta) )$.
\item
In the \cite{KyngS16} proof of Theorem~3.1, the variance
$\sigma_{3}^{2}$ is bounded by 
\begin{align*}
\sigma_{3}^{2} \leq
\sum_{\substack{\text{rounds} \\
   \text{ of elimination} \\
  i = 1 \text{ to } n-1}}
\norm{\boldsymbol{\Omega}_{i}}
=
\sum_{\substack{\text{rounds} \\
   \text{ of elimination} \\
  i = 1 \text{ to } n-1}}
 \frac{3}{\rho (n + 1 - i) }
\leq
 \frac{3\ln(n-1)}
{\rho}
.
\end{align*}
This bound ultimately relies on the $i^{th}$ vertex to eliminate being
chosen uniformly at random among$n + 1 - i$ vertices.
{\partialChol} only chooses vertices at random among the vertices of
the set $F$.
Thus the $i^{th}$ vertex to eliminate is chosen uniformly at random
among $\sizeof{F} + 1 - i$ vertices.
However, we also only make $\sizeof{F}$ eliminations, and ultimately,
the variance is bounded by
\begin{align*}
\sigma_{3}^{2} \leq
\sum_{\substack{\text{rounds} \\
   \text{ of elimination} \\
  i = 1 \text{ to } \sizeof{F} }}
\norm{\boldsymbol{\Omega}_{i}}
 = 
\sum_{\substack{\text{rounds} \\
   \text{ of elimination} \\
  i = 1 \text{ to } \sizeof{F} }}
\frac{3}{\rho ({F} + 1 - i) }
\leq
 \frac{3\ln(\sizeof{F} + 1)}
{\rho}
.
\end{align*}
As $\sizeof{F} < n$, we get $\ln(\sizeof{F} + 1)\leq \ln(n)$, and so
the variance $\sigma_{3}^{2}$ of {\partialChol} is less than the
corresponding variance of {\textsc{SparseCholesky}}.
Thus we are able to get the same concentration bounds for
{\partialChol} as for {\textsc{SparseCholesky}}.
\item To make the matrix of eliminated columns lower-triangular,
\textsc{SparseCholesky} permutes the rows (see \cite{KyngS16} \textsc{SparseCholesky}
algorithm Line 10).
{\partialChol} does not need the matrix of eliminated columns to be lower-triangular, so we
do not apply this permutation. 
This does not change the analysis in any way.
\item
The steps outlined above suffice to argue that
\[
\LL \approx_{\eps/2}
  \begin{pmatrix}
    \matlowtil_{FF} \\
    \matlowtil_{CF}
  \end{pmatrix}
  \calDtil
  \begin{pmatrix}
    \matlowtil_{FF} \\
    \matlowtil_{CF}
  \end{pmatrix}^{\top}
+ 
\begin{pmatrix}
\matzero_{FF} & \matzero_{FC} \\
\matzero_{CF} & \SStil^{(\sizeof{F})}
\end{pmatrix}
.
\]
and  $\SStil^{(\sizeof{F})}$ has
 $O(m\log n +
n \eps^{-2} \log n \log^{2} (n/\delta) )$ edges.

Finally, by Lemma~\ref{lem:sparsify}, setting 
 $\SStil \gets \sparsify(\SStil^{(\sizeof{F})},\eps, 2\delta)$
ensures that with probability $1-1/n^{2\delta}$
we get that
$\SStil$ has $O(k \eps^{-2} \log n )$ edges and 
$\SStil \approx_{\eps/2} \SStil^{(\sizeof{F})}$.
So by composing
guarantees 
\[
\LL \approx_{\eps}
  \begin{pmatrix}
    \matlowtil_{FF} \\
    \matlowtil_{CF}
  \end{pmatrix}
  \calDtil
  \begin{pmatrix}
    \matlowtil_{FF} \\
    \matlowtil_{CF}
  \end{pmatrix}^{\top}
+ 
\begin{pmatrix}
\matzero_{FF} & \matzero_{FC} \\
\matzero_{CF} & \SStil
\end{pmatrix}
.
\]

\end{enumerate}

We also need to check the overall running time of the algorithm:
The call to $\levEst$ takes time  $O(m \log^{2} (n/\delta)
\operatorname{polyloglog}(n) )$.
The elimination takes time 
$O(m\log n +
n \eps^{-2} \log n \log^{2} (n/\delta) )$.
The call to $\sparsify$ takes as input a graph with
 $O(m\log n +
n \eps^{-2} \log n \log^{2} (n/\delta) )$ edges and less than $n$ vertices,
and so it runs in time
 $O((m\log n \log^{2}(n/\delta)+
n \eps^{-2} \log n \log^{4} (n/\delta))
\operatorname{polyloglog}(n))$.
All together, the running time dominated by the $\sparsify$ call, so
it is  $O((m\log n \log^{2}(n/\delta)+
n \eps^{-2} \log n \log^{4} (n/\delta))
\operatorname{polyloglog}(n))$.

Finally, the $\levEst$ call, the elimination, and the $\sparsify$ call
each fail with probability $<\delta/3$, so the total failure
probability is less than $\delta$ by a union bound.
\end{proof}


\section{Effective Resistance Estimation}
\label{sec:reff-estimation}

Recalling the statement of Theorem~\ref{thm:esimateReff-analysis}, we will give our algorithm and show the following.

\estimateReffAnalysis*

First, we give our algorithm for estimating the effective resistance of a set of pairs $S$ that achieves an improved running time (ignoring $\log(n)$ factors) over algorithms that are based on the Johnson-Lindenstrauss Lemma, for a sufficiently small set of pairs and error parameter. The algorithm  \estimatereff (Algorithm~\ref{alg:estimateReff}) is given below. The main tool it uses is the ability to quickly compute a sparse spectral approximation of the Schur complement of a graph onto a subset of its vertices, along with the following observations:
\begin{enumerate}
\item An approximate Schur complement of a graph onto a subset $V_i$ of the vertices approximately preserves effective resistances between elements of $V_i$ 
\item If the number of vertex pairs we wish to compute the effective resistances of is much smaller than the number of vertices, then there must be a large number of vertices that are not part of any pair, and these vertices can be removed by taking a Schur complement, shrinking the size of the graph.
\end{enumerate}
Our proof of Theorem~\ref{thm:esimateReff-analysis} relies on Theorem~\ref{thm:combSchurApx}, which gives guarantees for the purely combinatorial Schur complement approximation algorithm~{\combScElim} that are almost as strong as the guarantees for the {\scElim} algorithm given in Theorem~\ref{thm:schurApx}.
We prove Theorem~\ref{thm:combSchurApx} in Section~\ref{sec:comblev}.

\begin{theorem}
\label{thm:combSchurApx}
  Given a connected undirected multi-graph
  $G =(V,E)$, with positive edges weights 
  $w : E \to \rea_{+}$, and associated Laplacian $\LL$,
  a set vertices $C \subset V$,
  and and scalars $0<\eps\leq1/2$, $0< \delta < 1$,
  the algorithm $\combScElim(\LL,C,\eps,\delta)$
  returns a Laplacian matrix $\SStil$.
  With probability $\geq 1-\delta$ the following statements hold:
  $\SStil \approx_{\eps} \SS$, where $\SS$ is the Schur complement of
  $\LL$ w.r.t elimination of $F=V-C$.
$\SStil$ is a Laplacian matrix whose edges are supported on $C$.
Let $k = \sizeof{C} = n-\sizeof{F}$.
The total number of non-zero entries $\SStil$ is $O(k
\eps^{-2}\operatorname{polylog}(n/\delta))$.
The total running time is bounded by
$O((m +
n \eps^{-2}) \operatorname{polylog}(n/\delta))$.
\end{theorem}

\begin{algorithm}[H]							
\caption{$\estimatereff(G=(V,E),S,\epsilon)$}
\label{alg:estimateReff}
\Indm									
\Input{A graph $G=(V,E)$, a set $S \subseteq V \times V$ of vertex pairs, and an error tolerance $0<\epsilon \leq 1$}
\Output{Estimates of the effective resistances of each of the pairs in $S$ accurate to within a factor of $e^{\pm \epsilon}$ with high probability}
\Indp

\vspace{1em}

$\epsilon' \gets \frac{\epsilon}{\log_2 n}$.

\Return{$\helpestimatereff(G,S,\epsilon')$}
\end{algorithm}

\begin{algorithm}[H]							
\caption{$\helpestimatereff(G=(V,E),S,\epsilon)$}
\label{alg:helpEstimateReff}
\Indm									
\Input{A graph $G=(V,E)$, a set $S \subseteq V \times V$ of vertex pairs, and an error tolerance $0<\epsilon \leq 1$}
\Output{Estimates of the effective resistances of each of the pairs in $S$ accurate to within a factor of $e^{\pm \epsilon \log_2 n}$ with high probability}
\Indp

\vspace{1em}

\If{$S = \emptyset$}{
\Return{$\emptyset$}
}

\vspace{1em}

Let $V_0$ denote the set of all vertices that are part of at least one pair in $S$.

$G \gets \combScElim\left(G,V_0,\epsilon,\text{with high probability}\right)$ \qquad (Algorithm in Theorem~\ref{thm:combSchurApx})

$V \gets V_0$

\vspace{1em}

\If{$|S| = 1$ (or equivalently, $|V|=2$)}{
Let $z$ denote the pair in $S$ or equivalently, the only two vertices in the graph.

\Return{the estimate $1/w_z$, where $w_z$ is the weight of the only edge in $G$.}
}

\vspace{1em}

Partition $V$ into $V_1,V_2$ with $|V_1|=\lfloor n/2 \rfloor$ and $|V_2|=\lfloor n/2 \rfloor$.

Partition $S$ into subsets $S_1,S_2,S_3$ with:

\Indp
\nonl$S_1\gets$ pairs with both elements in $V_1$

\nonl$S_2\gets$ pairs with both elements in $V_2$

\nonl$S_3\gets$ pairs with one element in $V_1$ and the other in $V_2$.

\Indm

Let $G_1\gets \combScElim\left(G,V_1,\epsilon,\text{with high probability}\right)$.

Let $G_2\gets \combScElim\left(G,V_2,\epsilon,\text{with high probability}\right)$.

Concatenate and return the estimates given by:

\Indp

\nonl$\helpestimatereff(G_1,S_1,\epsilon)$

\nonl$\helpestimatereff(G_2,S_2,\epsilon)$

\nonl$\helpestimatereff(G,S_3,\epsilon)$

\Indm
\end{algorithm}

\newcommand{\indvec}{\vec{1}}

We now prove that this algorithm quickly computes effective resistances. In doing this analysis, we did not try to optimize log factors, and we believe that at least some of them can likely be eliminated through a more careful martingale analysis.

\begin{proof}(of Theorem ~\ref{thm:esimateReff-analysis})
First we prove correctness. In any recursive call of \helpestimatereff (Algorithm~\ref{alg:helpEstimateReff}), let $\LL$ denote the Schur complement of the graph onto (say) $V_1$. Fact~\ref{fac:pinvandschur} says that the Schur complement of a graph onto a subset of its vertices $V_1$ exactly preserves effective resistances between vertices in $V_1$. However, the algorithm we are analyzing does not take an exact Schur complement. Instead, it takes an approximate Schur complement $\widetilde{\LL}$ which by Theorem~\ref{thm:combSchurApx}, satisfies $e^{-\epsilon'}\LL \preceq \widetilde{\LL} \preceq e^{\epsilon'} \LL$. We also know that the effective resistance between $i$ and $j$ in the approximate Schur complement is given by $(\indvec_i - \indvec_j)^\intercal \widetilde{\LL}^\dagger (\indvec_i - \indvec_j)$, where $\indvec_z$ is the $z$th standard basis vector. These two facts imply that the effective resistance between $i$ and $j$ in $\widetilde{\LL}$ is within an $e^{\pm \epsilon'}$ factor of what it was before taking the approximate Schur complement. Applying this inductively over the depth of the recursion, we get that the approximate effective resistances $\widetilde{R}_\text{eff}$ returned by the algorithm satisfy

\begin{align*}
e^{-\epsilon' (\lceil \log_2 n \rceil - 1)} R_\text{eff} &\leq \widetilde{R}_\text{eff} \leq e^{\epsilon' (\lceil \log_2 n \rceil - 1)} R_\text{eff} \\
e^{-\epsilon} R_\text{eff} &\leq \widetilde{R}_\text{eff} \leq e^{\epsilon} R_\text{eff} \\
\end{align*}

For runtime, let $n,m$ be the number of vertices and edges in the original graph, before any recursion is done. Consider any recursive call $c$. Let $n_c$ be the number of vertices of the graph $G$ that is given to $c$ as an argument, before any modifications within $c$ have been done. Let $s_c$ denote the number of pairs in the argument $S$ passed to the recursive call $c$. Finally, let $n'_c$ denote the number of vertices in $G$ after $G$ has been replaced with its Schur complement onto $V_0$ in the call. By Theorem~\ref{thm:combSchurApx}, the actual amount of work done in a recursive call of \helpestimatereff (other than the top level call) is $\widetilde{O}(n_c / \epsilon^2)$. Here and for the rest of this proof, $\widetilde{O}$ hides factors polylogarithmic in $n$, but does not hide anything that explicitly depends on on $n'_c$ or $\epsilon$.


We claim that with proper amortization, the amount of work done in each recursive call is $\widetilde{O}(n'_c/\epsilon^2)$. To show this, define a potential function $\phi_c$ which is $\widetilde{\Theta}(n'_c / \epsilon^2)$. Then define the amortized cost of a recursive call as its true cost plus $(\phi_c -\phi_{\text{parent}(c)}/3)$. Since the recursion tree has branching factor $3$, the sum of the amortized costs of the calls upper bounds the total true cost.

Then we have that the amortized cost of a call $c$ is 
\[
\widetilde{O}(n_c/\epsilon^2) + (\phi_c -\phi_{\text{parent}(c)}/3) = \widetilde{O}(n_c/\epsilon^2) + (\phi_c -\phi_{\text{parent}(c)})/3 + (2/3)\phi_c \leq \widetilde{O}(n_c'/\epsilon^2).
\]

Recall that $n'_c$ is the number of vertices in the graph given to the call that are part of at least one pair in $S$. Thus, $n'_c \leq 2s_c$. Putting this all together, we get that the total amortized work done in the first level of \helpestimatereff is $\widetilde{O}(m+n/\epsilon^2)$, and for any subsequent level, it is given by
\[
\sum_{\text{calls $c$ in the level}} \widetilde{O}(n_c/\epsilon^2) + (\phi_c -\phi_{\text{parent}(c)}/3) \leq \sum_{\text{calls $c$ in the level}} \widetilde{O}(s_c/\epsilon^2) \leq \widetilde{O}(|S|/\epsilon^2).
\]

Summing over all levels gives the claimed bound of
\[
\widetilde{O}\left(m + \frac{n+|S|}{\epsilon^2}\right).
\]
\end{proof}

\subsection{Combintorial Sparsification and Leverage Score Estimation}
\label{sec:comblev}

In this subsection we prove Theorem~\ref{thm:combSchurApx}, a version of Theorem~\ref{thm:schurApx}
that only uses combinatorial algorithms, at the expense of more
$\log$s in the running time and sparsity of the output.

The only non-combinatorial elements of element of $\scElim$ is the
calls to $\levEst$ and $\sparsify$, which both use
Johnson-Lindenstrauss based leverage score estimation.
Thus, the key to obtaining a combinatorial version of $\scElim$ is to
replace $\levEst$ and $\sparsify$ with purely combinatorial
counterparts that still have running times of the form
$O((m +n \eps^{-2}) \operatorname{polylog}(n/\delta))$
and produce leverage scores/sparse graphs with
$O(n \eps^{-2} \operatorname{polylog}(n/\delta))$ sum/edges
respectively, for a failure probability $\delta$.

We observe that sufficient components are already known in the
literature:
If we combine the sparsifier algorithm of \cite{KyngPPS17} (Theorem
4.1 with $\eps$ set to a constant) 
with the leverage score estimation algorithm of \cite{KoutisLP15}
(Lemma 6.5 which takes an arbitrary sparsifier),
to give a purely combinatorial leverage score estimation algorithm
$\combLevEst$, we immediately get the following result.
\begin{lemma}
\label{lem:comblevest}
 Given a connected undirected multi-graph
  $G =(V,E)$, with positive edges weights 
  $w : E \to \rea_{+}$, and associated Laplacian $\LL$,
  and a scalar $0 < \delta < 1$
  the algorithm $\combLevEst (\LL,\delta)$ returns estimates $\tauEst_e$ for all the edges such that
  with probability $\geq 1-\delta$
	\begin{enumerate}
		\item For each edge $e$, we have
		$\tau_{e} \leq \widehat{\tau}_e \leq 1$ where
		$\tau_{e}$ is the true leverage score of $e$ in $G$.
        \item $ \sum_{e} \tauEst_e
			\leq n \operatorname{polylog}(n).$
	\end{enumerate}
The algorithm runs in time $O(m \operatorname{polylog}(n/\delta)
)$.
\end{lemma}
If we then combine this with the sparsification of \cite{SpielmanS11},
get a combinatorial sparsification algorithm $\combSparsify$.
\begin{lemma}
\label{lem:combsparsify}
Given a connected undirected multi-graph
  $G =(V,E)$, with positive edges weights 
  $w : E \to \rea_{+}$, and associated Laplacian $\LL$,
and scalars $0<\eps\leq1/2$, $0< \delta < 1$,
$\combSparsify (\LL,\eps,\delta)$ returns a Laplacian $\LLtil$ s.t. 
with probability $\geq 1 - \delta$ it holds that
$\LLtil
\approx_{\eps} \LL$ and $\LLtil$ has $O(n \eps^{-2} \operatorname{polylog}(n/\delta) )$ edges.
The algorithm runs in time  $O( (m + n \eps^{-2}) \operatorname{polylog}(n/\delta) )$.
\end{lemma}
\begin{proof}(of Theorem~\ref{thm:combSchurApx})
If we replace $\levEst$ and $\sparsify$ in $\scElim$ with
$\combLevEst$ and $\combSparsify$ respectively and adjust parameters
appropriately, we then immediately get a purely combinatorial
algorithm $\combScElim$ for Schur complement approximation, proving
the theorem.
\end{proof}

\section*{Acknowledgements}

We thank Richard Peng for extensive discussions and comments. We thank Michael Cohen, who independently observed that the running time of our algorithm could be improved by changing the error parameters of the algorithm, as we have done in this version of the paper.

\bibliographystyle{alpha}
\bibliography{ref}

\end{document}